\documentclass{article}[11pt]

\usepackage{booktabs} 
\usepackage{graphicx}
\usepackage{graphics}
\usepackage{epsfig}
\usepackage{wrapfig}
\usepackage{latexsym}
\usepackage{amsmath} 
\usepackage{stmaryrd}
\usepackage{ifthen}
\usepackage{array}
\usepackage{mathrsfs}
\usepackage{paralist}
\usepackage{tikz}
\usetikzlibrary{arrows, automata, positioning,calc}
\renewcommand{\figurename}{Figure}
\usepackage{pgf, verbatim,eurosym}
\usetikzlibrary{arrows,automata}
\usepackage{sansmath}
\usepackage{dsfont}
\usepackage{enumerate}
\usepackage{amssymb}
\usepackage{amsmath,amsthm}

\newtheorem{theorem}{Theorem}
\newtheorem{lemma}[theorem]{Lemma}

\newcommand{\true}{\mathsf{true}}
\newcommand{\A}{\mathcal A}

\newcommand{\R}{\mathcal R}
\newcommand{\locs}{L}
\newcommand{\loc}{\ell}

\newcommand{\ttto}[3]{\xrightarrow{#1 \,\, \, #2 \,\,\, #3\downarrow}}
\newcommand{\tto}[2]{\xrightarrow{#1 \,\, \, #2 \,\,}}
\newcommand{\Hsquare}{\text{\fboxsep=-.2pt\fbox{\rule{0pt}{1ex}\rule{1ex}{0pt}}}}

\newcommand{\domain}{D}
\newcommand{\alphabet}{\Sigma}
\newcommand{\reg}{R}
\newcommand{\up}{\mathsf{up}}
\newcommand{\valuation}{\nu}

\newcommand{\nat}{\mathbb{N}}
\newcommand{\tuple}[1]{\langle #1 \rangle}

\def\abs#1{\ensuremath{\lvert #1\rvert}}

\newcommand{\semantics}[1]{\llbracket #1\rrbracket}
\newcommand{\post}{\mathsf{post}}

\newcommand{\wait}{\mathsf{wait}}

\RequirePackage{tikz}
\usetikzlibrary{arrows,snakes}

\makeatletter
\def\@tvsp{\mathchoice{{}\mkern-5.5mu}{{}\mkern-5.5mu}{{}\mkern-3.5mu}{}}
\def\ltrivert{\left\langle\@tvsp\left\langle}
\def\rtrivert{\right\rangle\@tvsp\right\rangle}
\makeatother

\newcommand{\defequals}{\stackrel{\mathrm{def}}{=}}
  
\newcommand*\ie{\textit{i.e.}}  
\newcommand*\eg{\textit{e.g.}}

\newcommand{\val}{\valuation}

\makeatletter
\newcommand*{\defeq}{\mathrel{\rlap{%
                     \raisebox{0.3ex}{$\m@th\cdot$}}%
                     \raisebox{-0.3ex}{$\m@th\cdot$}}%
                     =}
\makeatother

\makeatletter
\newcommand*{\ndefeq}{\mathrel{\rlap{%
                     \raisebox{0.3ex}{$\m@th\cdot$}}%
                     \raisebox{-0.3ex}{$\m@th\cdot$}%
                     \rlap{%
                     \raisebox{0.3ex}{$\m@th\cdot$}}
                     \raisebox{-0.3ex}{$\m@th\cdot$}}
                     =}
\makeatother

\newcommand{\data}{\mathsf{data}} 

\newcommand{\synch}{\mathsf{synch}}
\newcommand{\init}{\mathsf{init}}

\newcommand{\nlogspace}{\textup{NLOGSPACE}}
\newcommand{\pspace}{\textup{PSPACE}}
\newcommand{\conpspace}{\textup{co\text{-}(N)PSPACE}}
\newcommand{\npspace}{\textup{(N)PSPACE}}
\newcommand{\ackermann}{\textup{Ackermann}}

\newcommand{\exptime}{\textup{EXPTIME}}
\newcommand{\nexptime}{\textup{NEXPTIME}}
\newcommand{\conexptime}{\textup{co\text{-}NEXPTIME}}

\newcommand{\tower}{\textup{tower}}

\newenvironment{qlemma}[1]{%
{\medskip\noindent{\bf Lemma #1.}\hspace{-2mm}}%
\begin{itshape}%
}{%
\end{itshape}%
}

\newcommand{\lzero}{\mathsf{zero}}

\newcommand{\lreset}{\mathsf{reset}}

\newcommand{\inBinary}[1]{\texttt{{#1}}}
\newcommand{\Bit}[1]{\mathsf{Bit}_{#1}}
\newcommand{\loctwo}[1]{2^{#1}}
\newcommand{\locctwo}[1]{2^{#1}_c}

\newcommand{\counter}{\mathsf{counter}}

\newcommand{\waitexp}{\mathsf{waitExp}}
\newcommand{\waittow}{\mathsf{waitTow}}
\newcommand{\waitdoub}{\mathsf{waitDoub}}
\newcommand{\waitrep}{\mathsf{waitRep}}
\newcommand{\rep}{\mathsf{rep}}
\newcommand{\store}{\mathsf{store}}
\newcommand{\arepeat}{\mathsf{rep}}
\newcommand{\atower}{\mathsf{tow}}
\newcommand{\adouble}{\mathsf{doub}}
\newcommand{\aexp}{\mathsf{exp}}
\newcommand{\Data}{\mathsf{data}}
\newcommand{\lang}{\mathsf{lang}}
\newcommand{\proj}{\mathsf{proj}}

\newcommand{\theokDRA}{
	The synchronization problem for $k$-DRAs is $\pspace$-complete. 
}

 \newcommand{\is}{\textup{in}} 
 \newcommand{\acc}{\textup{f}} 
 \newcommand{\Rfamily}{\mathbf{R}}

 \newcommand{\tm}{\mathcal{M}}

\advance \topmargin by -\headheight
\advance \topmargin by -\headsep
\evensidemargin \oddsidemargin
\begin{document}
\title{Synchronizing Data Words for Register Automata} 
\author{Karin Quaas \\
Universit\"at Leipzig
\and
Mahsa Shirmohammadi\\
CNRS \& IRIF
}

\date{}
\maketitle

\begin{abstract}
Register automata (RAs) are  finite  automata extended with a finite set of registers
 to store and compare data from an infinite domain. 
We study  the concept of  synchronizing data words in RAs: 
does there exist a data word that sends all states of the RA
to a single state?

For deterministic RAs with $k$ registers ($k$-DRAs), 
we prove that inputting data words with  $2k+1$ distinct data from the infinite data domain 
is  sufficient to synchronize. 
We show that the synchronization problem for DRAs is in general $\pspace$-complete,
and it is  $\nlogspace$-complete for $1$-DRAs. 
For nondeterministic RAs (NRAs), we show that
$\ackermann(n)$ distinct data (where $n$ is the size of the RA)
might be necessary to synchronize.
The  synchronization problem for NRAs is in general undecidable, however,
we   establish $\ackermann$-completeness of  
the   problem for $1$-NRAs. 
Another main result is the  $\nexptime$-completeness of 
the length-bounded synchronization problem for NRAs, where a bound on the length of the synchronizing data word, written in binary, is given.
A variant of this  last construction allows to prove that the length-bounded 
universality problem for NRAs is $\conexptime$-complete.
\end{abstract}

\maketitle
\section{Introduction}
Given a deterministic finite automaton (DFA), a \emph{synchronizing word} is a word that sends all states of the automaton
to a unique state.  Synchronizing words for finite automata have been studied since the 1970\emph{s}~\cite{Cerny64,Pin78,Volkov08,M10}
and are the subject of one of the most well known open problems in automata theory---the 
\v{C}ern\'{y} conjecture.
This conjecture states that the length of a shortest synchronizing word
for a DFA with $n$ states is at most~$(n-1)^{2}$.
Synchronizing words moreover have applications in planning, control of discrete event systems, biocomputing, and
robotics~\cite{Ben03,Volkov08,DMS11b}.  More recently the notion has been
generalized from automata to  games~\cite{LarsenLS14,MahsaThesis,KretinskyLLS15} and  
infinite-state systems~\cite{DBLP:conf/fsttcs/0001JLMS14,MahsaDmitry16}, with applications to modelling complex 
systems such as distributed data networks or real-time embedded systems.

In this paper we are interested in  \emph{synchronizing data words} for \emph{register automata}.
\emph{Data words} are  sequences of pairs, where the first element of each pair is taken from a finite alphabet and the second element 
is taken from \emph{an infinite data domain}, such as the natural numbers or ASCII strings.  
Data words have 
applications in querying and reasoning about data models with complex structural properties, e.g., XML
and graph databases~\cite{Angles:2008:SGD:1322432.1322433,Figueira:2009:SDX:1559795.1559827, Bojanczyk:2011:XEL:1989727.1989731,Barcelo:2012:ELP:2389241.2389250}. For reasoning about data words,
various formalisms have been considered, including
first-order logic for data words~\cite{DBLP:conf/lics/BojanczykMSSD06,DBLP:conf/fct/BouajjaniHJS07},
extensions of linear temporal logic~\cite{DBLP:conf/time/LisitsaP05,DBLP:journals/iandc/DemriLN07,DBLP:journals/tocl/DemriL09,DBLP:journals/tcs/DemriLS10}, data automata~\cite{DBLP:journals/iandc/BouyerPT03,DBLP:conf/lics/BojanczykMSSD06}, register automata~\cite{DBLP:journals/tcs/KaminskiF94,DBLP:journals/tcs/SakamotoI00,DBLP:journals/tocl/NevenSV04,DBLP:journals/tocl/DemriL09} and extensions thereof, \eg \,~\cite{DBLP:conf/popl/Tzevelekos11,DBLP:journals/corr/abs-1202-3957,ClementeL15}.

\emph{Register automata} (RAs) are a generalization of finite automata for processing data words. 
RAs are equipped with a finite set of registers that can store data values.
While processing a data word 
such an automaton can store the datum at the current position in one of its registers; it can also test the current datum for equality with data already stored in its registers.
In applications, RAs allow for handling parameters such as user names, passwords, identifiers of connections, sessions, etc.
RAs come in many variants, including one-way, two-way, deterministic, nondeterministic, and alternating.
For alternating one-way RAs, classical language-theoretic decision problems, such as emptiness, universality and inclusion are undecidable. 
In this paper, we focus on the class of one-way nondeterministic RAs, which have
a decidable emptiness problem~\cite{DBLP:journals/tcs/KaminskiF94},
and the subclass of nondeterministic RAs with a single register, which has a decidable universality problem~\cite{DBLP:journals/tocl/DemriL09}.

Semantically, an RA defines an infinite-state transition system due to the unbounded domain for the data stored in the registers.
Synchronizing words were introduced for infinite-state systems with infinite branching in~\cite{DBLP:conf/fsttcs/0001JLMS14,MahsaThesis};
in particular, the notion of synchronizing words is motivated and studied for weighted automata and  timed automata.
In some infinite-state settings, such as  nested-word automata,
finding the right definition of synchronizing  word is however more challenging~\cite{MahsaDmitry16}. 
We  define the \emph{synchronization problem} for RAs within the framework suggested  in~\cite{DBLP:conf/fsttcs/0001JLMS14,MahsaThesis}: 
given an RA~$\R$ over a finite alphabet~$\Sigma$ and an infinite data domain $\domain$, 
does there exist a  data word~$w \in (\Sigma\times \domain)^{+}$ and
some state $q_w$ such that the word~$w$ sends each of the infinitely many states  of $\R$ to~$q_w$? Note that the state $q_w$ depends on the word~$w$;
we call such a data word a \emph{synchronizing data word}.

\emph{Contribution.}
The problem of finding synchronizing data words for RAs poses new challenges in the area of  synchronization.
It is  natural to ask how many distinct data are necessary and sufficient to synchronize an RA, 
which we refer to as the  \emph{data efficiency} of synchronizing data words.
We show that the data efficiency is polynomial in the number of registers for deterministic RAs (DRAs). 
For nondeterministic RAs (NRAs), we provide an example that shows that the data efficiency may be $\ackermann(n)$, where $n$~is the number of states of the NRA. 
Remarkably, the data efficiency is tightly related to the complexity of deciding the synchronization problem. 
For DRAs,
we prove that for all automata~$\R$ with $k$ registers, 
if $\R$ has a synchronizing data word, 
then it also has one with data efficiency \emph{at most} $2k+1$.
We  provide a family $(\R_k)_{k\in\nat}$ of DRAs with $k$ registers, 
for which indeed a polynomial data efficiency (in $k$) is necessary to  synchronize. 
This bound is the base of an $\npspace$-algorithm for DRAs;
we prove a matching   $\pspace$ lower bound  by ideas carried over from timed settings~\cite{DBLP:conf/fsttcs/0001JLMS14}.
We show  that the synchronization problems for DRAs with a single register ($1$-DRAs) and for DFAs are $\nlogspace$-interreducible,  implying that
the problem is $\nlogspace$-complete for $1$-DRAs.

For NRAs, a reduction from the non-universality problem yields the undecidability of the synchronization problem.
For single-register NRAs ($1$-NRAs), we prove $\ackermann$-completeness of the problem  by 
a novel construction proving that  the synchronization problem and the non-universality problem for $1$-NRAs are  
polynomial-time interreducible.
We believe that this technique is useful in studying synchronization in all nondeterministic settings, requiring careful analysis of
the size of the construction.

Another main contribution is to prove  $\nexptime$-completeness of 
the \emph{length-bounded synchronization problem} for NRAs:
given a  bound on the length (written in binary), 
does there exist a synchronizing data word with length at most the given bound?
For the lower bound, we present a  reduction from the 
membership problem of  $\mathcal{O}(2^n)$-time bounded nondeterministic Turing machines.
The crucial ingredient in this reduction is a family of RAs implementing binary counters.
A variant of our construction yields a proof for $\conexptime$-completeness of the \emph{length-bounded universality problem} for NRAs; 
the length-bounded universality problem  asks whether all data words 
of length at most a given bound (written in binary) are in the language of the automaton.
We further  make a connection to the emptiness problem of single-register alternating RAs.

An extended abstract of this article has appeared in the Proceedings 
of the 41st International Symposium on Mathematical Foundations of Computer
Science, (MFCS) 2016~\cite{BQS16}.
In comparison with the extended abstract, 
here we simplify two of the main constructions and add detailed proofs of all results.
The main improvement is giving a simpler  $\nexptime$-hardness reduction 
for the length-bounded synchronization problem for NRAs.

\section{Preliminaries}
A deterministic finite-state automaton (DFA) is a  tuple~$\A=\tuple{Q,\alphabet,\Delta}$, 
where~$Q$ is a finite set of states, 
$\alphabet$ is a finite alphabet, and~$\Delta:Q\times \alphabet \to Q$
is a transition function that is totally defined.  
The  function~$\Delta$  extends to finite words in a natural way: 
$\Delta(q,wa)=\Delta(\Delta(q,w),a)$ for all words~$w\in \alphabet^{*}$ 
and letters~$a\in \alphabet$; it extends to all sets~$S\subseteq Q$  by $\Delta(S,w)=\bigcup_{q\in S} \Delta(q,w)$. 

\medskip

{\bf Data Words and Register automata.} 
For the rest of this paper, fix an infinite data domain~$\domain$. 
Given a finite alphabet $\alphabet$, a \emph{data word over $\alphabet$} is a  finite words over~$\alphabet\times\domain$. 
For a data word $w=(a_1, d_1) (a_2,d_2) \cdots (a_n,d_n)$,  
the length $\abs{w}$ of~$w$ is $n$. 
We use $\data(w)=\{d_1,\dots,d_n\}\subseteq \domain$ to refer to the set of data values occurring in~$w$, 
and we define the \emph{data efficiency of $w$} to be~$\abs{\data(w)}$. 

Let $\reg$ be a finite set of \emph{register variables}.
We define \emph{register constraints $\phi$ over $\reg$} by the grammar 
\[
\phi \; ::= \true \; | \; =\!r \; |   \; \phi \wedge \phi \; | \; \lnot \phi,\] where~$r \in \reg$. 
We denote by~$\Phi(\reg)$ the set of all register constraints over~$\reg$. 
We may use $\neq\!r$ for  the inequality constraint~$\neg(=\! r)$. 
A \emph{register valuation} is a mapping $\valuation: \reg \to \domain$
that assigns a data value to each register; 
we sometimes write~$\valuation= (\valuation(r_1), \cdots,  \valuation(r_k)) \in \domain^{k}$,   
 where~$\reg=\{r_1,\cdots, r_k\}$. 
The satisfaction relation of register constraints is defined  on $\domain^{k}\times \domain$ as follows:
$(\valuation,d)$ satisfies the constraint $=\!r$ if $\val(r)=\! d$; 
the other cases follow. 
For example,  $((d_1,d_2,d_1),d_2)$  satisfies $((=r_1)\wedge (=r_2)) \vee (\neq r_3))$ if $d_1\neq d_2$.
For the set~$\up\subseteq\reg$ and $d\in\domain$, we define the \emph{update $\valuation[\up \defeq d]$ of valuation $\valuation$} 
 by~$(\valuation[\up \defeq d])(r) = d$ if $r \in \up$, 
and $(\valuation[\up \defeq d])(r) = \valuation(r)$ otherwise.

A {\em register automaton} (RA) is a tuple $\R =\tuple{\locs,\reg,\alphabet,T}$,  
where $\locs$ is a finite set of locations, 
$\reg$ is a finite set of registers, 
$\alphabet$~is a finite alphabet and 
$T \subseteq \locs \times \alphabet \times \Phi(\reg) \times 2^{\reg} \times \locs$ is a transition relation. 
We may use~$\loc \ttto{\phi}{a}{\up} \loc'$ 
to show transitions $(\loc, a, \phi, \up, \loc') \in T$.  
We call $\loc \ttto{\phi}{a}{\up}\loc'$ an $a$-transition and $\phi$ the \emph{guard} of this transition.
A guard $\true$ is vacuously true and may be omitted. Likewise we may omit $\up$ if $\up=\emptyset$.
We may  write $r\!\downarrow$ when $\up=\{r\}$ is  a singleton set. 
For NRAs with only one register, we may shortly write $=$ and $\neq$ for the guards $=\!r$ and $\neq\!r$, respectively, and $\downarrow$ for the update $\downarrow\!r$.

A \emph{configuration} of~$\R$ is a  pair~$(\loc,\valuation) \in \locs \times \domain^{\abs{\reg}}$ of a location~$\loc$ and a register valuation~$\valuation$.
We describe the behaviour of~$\R$ as follows:   
Given a configuration~$q=(\loc,\valuation)$ and some input $(a,d)\in \alphabet\times\domain$
an $a$-transition $\loc \ttto{\phi}{a}{\up} \loc'$ may be fired from $q$ if
$(\valuation,d)$ satisfies the constraint~$\phi$; 
then~$\R$ moves to the \emph{successor configuration}~$q'=(\loc',\valuation')$, where $\valuation'=\valuation[\up \defeq d]$
is the update of $\valuation$. 
By~$\post(q, (a,d))$, we denote the set of all successor configurations~$q'$ of~$q$ on input $(a,d)$.
We extend $\post$  to sets~$S\subseteq L\times\domain^{\abs{\reg}}$ of configurations by $\post(S,(a,d))=\bigcup_{q\in S} \post(q,(a,d))$;
and we extend $\post$ to words by
$\post(S,w\cdot(a,d))=\post(\post(S,w),(a,d))$ for all words~$w\in (\alphabet\times \domain)^{*}$, and all inputs $(a,d)\in \alphabet\times \domain$.

A run of~$\R$ over the data word~$w=(a_1, d_1) (a_2,d_2) \cdots (a_n, d_n)$ 
is a sequence of configurations $q_0 q_1\dots q_n$, where $q_i\in \post(q_{i-1},(a_i, d_i))$ for all $1\leq i \leq n$. 
If $\R$ reaches a configuration $q=(\loc,x)$ during processing a word~$w$,  
we may say that \emph{an $x$-token is in $\loc$} (or simply \emph{a token is in $\loc$}).

In the rest of the paper, we consider {\em complete} RAs, 
meaning that for all configurations~$q\in \locs \times \domain^{\abs{\reg}}$ 
and all inputs~$(a,d)\in \alphabet \times \domain$, there is 
at least one successor: $\abs{\post(q,(a,d))}\geq 1$.
We also classify the RAs into  {\em deterministic} RAs (DRAs) and {\em nondeterministic} (NRAs),
where an RA is deterministic if $\abs{\post(q,(a,d))}\leq 1$ for all configurations~$q$ and all inputs~$(a,d)$. 
A \emph{$k$-NRA} (\emph{$k$-DRA}, respectively) is an NRA (DRA, respectively) with $\abs{\reg}=k$.

{\bf Synchronizing words and synchronizing  data words.}
\emph{Synchronizing  words} are a well-studied concept for DFAs,  see, \eg,~\cite{Volkov08}.
Informally,
a synchronizing word leads the automaton from every state to the same state. Formally, 
the word~$w\in \alphabet^{+}$ is synchronizing  for a DFA~$\A=\tuple{Q,\alphabet,\Delta}$ if there 
exists some state~$q\in Q$ such that~$\Delta(Q,w)=\{q\}$.
The \emph{synchronization problem} for DFAs asks, given a DFA~$\A$, whether 
there exists some synchronizing word for~$\A$.

The synchronization problem for DFAs  is in $\nlogspace$ 
by using the \emph{pairwise synchronization} technique:
given a DFA~$\A=\tuple{Q,\alphabet,\Delta} $, it is known that $\A$ has a synchronizing word
if and only if for all pairs of states~$q,q'\in Q$, there exists a word~$v$
such that~$\Delta(q,v)=\Delta(q',v)$ (see~\cite{Volkov08} for more details).
The pairwise synchronization algorithm initially sets $S_{\abs{Q}}=Q$. 
For  $i=\abs{Q}-1,\cdots,1$, the algorithm repeats the following two steps: 
(a) For two distinct states $q,q'\in S_{i+1}$, 
 find $v_i$  such that~$\Delta(q,v_i)=\Delta(q',v_i)$.
(b) Set  $S_{i}=\Delta(S_{i+1},v_i)$ (and repeat the loop). 
The word~$w=v_{\abs{Q}-1} \cdots  v_{1}$ is synchronizing for $\A$.

We introduce synchronizing data words for  RAs.
Given an RA~$\R = \tuple{\locs, \reg,\alphabet, T}$, 
a data word~$w \in (\alphabet\times\domain)^+$ is  \emph{synchronizing} for $\R$ if there exists some
configuration~$q_w=(\loc, \valuation)$ such that 
$\post(\locs \times \domain^{\abs{\reg}} ,w)=\{q_w\}$.
Intuitively, no matter what is the starting location and register valuation,
by inputting the data word~$w$, $\R$ will be in the unique successor configuration~$q_w$. 
This configuration $q_w$ depends on~$w$.
The \emph{synchronization problem} for RAs asks, given an RA~$\R$ over a data domain~$\domain$,
whether there exists some synchronizing data word for~$\R$. 
The \emph{length-bounded synchronization problem} for RAs  decides, given an RA~$\R$ and a bound~$N\in \nat$ written in binary, whether there exists some 
synchronizing data word~$w$ for $\R$ satisfying $\abs{w}\leq N$. 

\section{Synchronizing data words for DRA\lowercase{s}}\label{synchDRAs}
In this section, 
we first show that   the synchronization problems for $1$-DRAs and DFAs are $\nlogspace$-interreducible,  implying that
the problem is $\nlogspace$-complete for $1$-DRAs. 
Next, we prove that the problem for $k$-DRAs, in general, can be decided in $\pspace$;  
a reduction similar to a timed setting, 
as in~\cite{DBLP:conf/fsttcs/0001JLMS14},  provides  the matching lower bound.
To obtain the complexity upper bounds, we prove that inputting words with data efficiency~${2\abs{\reg}+1}$ is sufficient to synchronize
a DRA.

The concept of synchronization requires that all runs of an RA, whatever the initial configuration (initial location and register valuations),
end in the same configuration~$(\loc_\synch,\nu_\synch)$, only depending on the synchronizing data word~$w_{\synch}$, formally 
$\post(\locs \times \domain^{\abs{\reg}},w_{\synch})=\{(\loc_\synch,\nu_\synch)\}$.
While processing a synchronizing data word, the infinite set of configurations of RAs must necessarily shrink to a finite set of configurations. 
The DRA~$\R$ with~$3$ registers depicted in \figurename~\ref{fig:UpdateEfficienyRegisterAutomata}
illustrates this phenomenon. 
Consider the set~$\{x_1,x_2,x_3\}\subseteq \domain$ of distinct data values: 
starting from any of the infinite configurations in~$\{\init\}\times \domain^3$, 
when processing the data word~$(a,x_1)(a,x_2)(a, x_3)$, 
$\R$ will be in a configuration in the finite set~$\{(\loc_3,(x_1, x_2, x_3)), (\loc'_3,(x_1, x_2, x_3)\}$.
We use this observation to provide a linear bound on the number of distinct data values that is sufficient for synchronizing DRAs.

In Lemma~\ref{lemmafiniteD} below, we prove  that data words over only $\abs{\reg}$ distinct data values are sufficient to 
shrink the infinite set of all configurations of DRAs to a finite set. 
We establish this result based on the following two key facts:

($1$)
When processing a synchronizing data word $w_\synch$ from a configuration 
$(\loc,\nu)$ with some register $r\in\reg$ such that  $\nu(r)\not\in\data(w_{\synch})$, 
the register $r$ must be updated.
Observe that such updates must happen at inequality-guarded transitions, which themselves
must be accessible by inequality-guarded transitions (possibly with no update). 
As an example, consider the  DRA~$\R$ in \figurename~\ref{fig:UpdateEfficienyRegisterAutomata}, and assume  
$d_1,d_2\not\in  \data(w_{\synch})$. 
The two runs of $\R$ starting from~$(\init,d_1,d_1,d_1)$
and~$(\init,d_2,d_2,d_2)$ first
  take the transition $\init \ttto{\neq r_1}{a}{r_1} \loc'_1$ updating register~$r_1$. 
  Next, the two runs must take~$\loc'_1 \ttto{\text{{\sf else}}}{a}{r_{2}} \loc'_{2}$ to update~$r_2$ and
$\loc'_2 \ttto{\text{{\sf else}}}{a}{r_{3}} \loc'_{3}$ to update~$r_3$; otherwise these two runs would never be synchronized in a single configuration.

($2$) Moreover, 
 to shrink the set~$\locs \times \domain^{\abs{\reg}}$, for every $\loc\in\locs$, 
one can find a word~$w_{\loc}$ that leads the DRA 
 from~$\{\loc\}\times \domain^{\abs{\reg}}$ to some finite set. 
Since $\R$ is deterministic, appending some prefix or suffix 
 to~$w_{\loc}$  achieves the same objective. This allows us to use a variant of the \emph{pairwise synchronization} technique to 
shrink the infinite set $\locs \times \domain^{\abs{\reg}}$ to a finite set, by successively inputting $w_{\loc}$
 for a location~$\ell$ that appears with infinitely many data in the current successor set of~$\locs \times \domain^{\abs{\reg}}$.

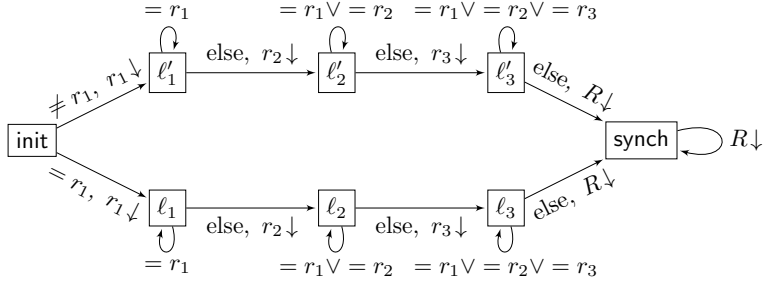
\begin{figure}[t]
\begin{center}
\scalebox{.9}{
    
\begin{tikzpicture}[>=latex',shorten >=1pt,node distance=2.3cm,on grid,auto,
roundnode/.style={circle, draw=black, minimum size=2mm},
squarednode/.style={rectangle, draw=black,  minimum size=2mm}]

\node[squarednode] (1) at (0,0) {$\init$};
\node[squarednode]  (2) [right =2cm of 1,yshift=-1cm]  {$\ell_{1}$};
\node[squarednode]  (3) [right =2cm of 1,yshift=1cm]  {$\ell'_{1}$};

\node[squarednode] (4) [right =2.5cm of 2] {$\ell_{2}$};
\node[squarednode] (5) [right =2.5cm of 4] {$\ell_{3}$};

\node[squarednode] (6) [right =9cm of 1] {$\synch$};
\node[squarednode] (7) [right =2.5cm of 3] {$\ell'_{2}$};
\node[squarednode] (8) [right =2.5cm of 7] {$\ell'_{3}$};
     

\path[->] (2) edge [loop below]  node [below] {$=r_{1}$} ();
\path[->] (4) edge [loop below]  node [below] {$=r_{1}\vee =r_{2}$} ();
\path[->] (5) edge [loop below]  node [below] {$=r_{1}\vee =r_{2}\vee =r_{3}$} ();
\path[->] (3) edge [loop above]  node [above] {$=r_{1}$} ();
\path[->] (7) edge [loop above]  node [above] {$=r_{1}\vee =r_{2}$} ();
\path[->] (8) edge [loop above]  node [above] {$=r_{1}\vee =r_{2}\vee =r_{3}$} ();


\path[->] (1) edge node [midway,sloped,below] {$=r_{1},~r_{1}\!\downarrow$}  (2);

\path[->] (2) edge node [midway,sloped,below] {$\text{else},~r_{2}\!\downarrow$}  (4);
\path[->] (4) edge node [midway,sloped,below] {$\text{else},~r_{3}\!\downarrow$}  (5);
\path[->] (5) edge node [midway,below,sloped] {$\text{else},~\reg\!\downarrow$}  (6);

\path[->] (3) edge node [midway,sloped,above] {$\text{else},~r_{2}\!\downarrow$}  (7);
\path[->] (7) edge node [midway,sloped,above] {$\text{else},~r_{3}\!\downarrow$}  (8);
\path[->] (8) edge node [midway,above,sloped] {$\text{else},~\reg\!\downarrow$}  (6);

\path[->] (6) edge [loop right]  node [right] {$\reg \!\downarrow$} ();

\path[->] (1) edge node [midway,sloped,above] {$\neq r_{1},~r_{1}\!\downarrow$}  (3);
  \end{tikzpicture}
 }
\end{center}
 \caption{A DRA with registers~$r_1,r_2,r_3$ and the single letter~$a$ (omitted from transitions) that can be synchronized in
the configuration~$(\synch, x_4)$ by the data 
word~$w_{\synch}=(a,x_1)(a,x_2)(a, x_3)(a, x_4)$ if 
$\{x_1, x_2, x_3, x_4\}\subseteq \domain$ 
is a set of 4 distinct data.} 
\label{fig:UpdateEfficienyRegisterAutomata}
\end{figure}

\newcommand{\lemmafiniteD}{
	For all DRAs for which there exist synchronizing data words, 
	there exists some data word~$w$ 
	such that $\data(w) \leq \abs{\reg}$ and $\post(\locs\times \domain^{\abs{\reg}}, w)\subseteq \locs \times (\data(w))^{\abs{\reg}}$.
}
\begin{lemma} \label{lemmafiniteD}
\lemmafiniteD
\end{lemma}
\begin{proof}
Let $\R =\tuple{\locs,\reg,\alphabet,T}$ be a DRA on the data domain~$\domain$ with $k\geq 1$ registers. 
Let
$v$ be a synchronizing data word for~$\R$ with $N=\abs{\data(v)}$ distinct data.
Suppose that $k<N$; otherwise the statement of the lemma trivially holds.

For all $1\leq i \leq k$, we say that \emph{$x_i$ is the $i$-th datum in the synchronizing data word~$v=(a_1,d_1)(a_2,d_2)\cdots (a_n,d_n)$}
 if there exists~$ j \leq k$ such that
$x_i=d_j$,  
$x_i\not \in \{d_1,\cdots,d_{j-1}\}$ and $\abs{\{d_1,\cdots,d_{j}\}}=i$. 
	For every $i\leq k$, denote by  $\langle L,i\rangle$  the set
$$\langle L,i\rangle =L \times \{\val\in \domain^k\mid \exists \reg'\subseteq \reg \, \cdot \, \abs{R'}\geq i \, \cdot \, \forall r\in \reg' \cdot
 \val(r)\in \{x_1,\cdots,x_i\}\}.$$

We {\bf Claim} that for  all locations~$\ell \in \locs$ and all~$1 \leq i \leq k$, there exists some data word~$u_i$ such that
\begin{itemize}
	\item $\data(u_i) \subseteq \{x_1,x_2,\cdots,x_i\}$, and
	\item 
		$\post(\{\ell\} \times\domain^k,u_i)\subseteq  \langle \locs,i\rangle$, 
	meaning that after reading~$u_i$ all reached configurations have at least~$i$ registers with values from~$\{x_1,x_2,\cdots,x_i\}$.
\end{itemize}
For $\loc \in \locs$, let $w_{\ell}=u_k$ satisfy the above condition. 
Set $S_0 = \locs \times \domain^k$ and $w_0=\varepsilon$. 
Then, for all $i=1,\cdots,\abs{\locs}$, repeat the following: 
if there exists some $\loc\in \locs$ such that $\{\loc\}\times (\domain \setminus \{x_1,\cdots,x_k\})^k \cap S_{i-1} \neq \emptyset$,  
then set $w_i = w_\loc$ and $S_i=\post(S_{i-1},w_i)$. Otherwise set $w_i=w_{i-1}$ and $S_i=S_{i-1}$.
Observe that $w=(w_i)_{1\leq i \leq \abs{\locs}}$  proves the statement of Lemma. It remains to prove the {\bf Claim}.

 \medskip

\noindent {\bf Proof of Claim.}  Let $\hat{\ell}$ be some location in the DRA~$\R$. The proof is by an induction on~$i$.

\noindent {\bf Base of induction.}
Let  $\wait=\{\hat{\loc}\} \times (\domain\setminus \data(v))^{k}$ be the set of  
configurations with location~$\hat{\loc}$ such that the data stored in all $k$~registers
is not in~$\data(v)$. 
 Note that for all configurations $(\hat\loc,\val)\in\wait$, 
the unique run of~$\R$ starting in $(\hat\loc,\val)$ on (a prefix of)~$v$ consists of the same sequence of the 
following transitions: 
\begin{itemize}
\item a prefix of transitions $\ttto{\bigwedge_{r\in \reg}\neq r}{}{\emptyset}$, 
 with inequality guards on all registers and with no register update,

\item followed by a transition $\ttto{ \bigwedge_{r\in \reg} \neq r}{}{\up}$, with inequality guard on all registers 
and with an update  for some non-empty set $\up\subseteq \reg$.
\end{itemize}

Otherwise, the two runs starting from any pair of configurations $(\hat{\loc},\valuation_1),(\hat{\loc},\valuation_2)\in \wait$
 with unequal valuations~$\valuation_1\neq \valuation_2$ would end up in distinct configurations,
say $(\loc,\valuation'_1),(\loc,\valuation'_2)$ with $\valuation'_1\neq \valuation'_2$.
This is a contradiction to the fact that the data word $v$ is  synchronizing.

Now let the inequality-guarded transition~$\ttto{\bigwedge_{r\in \reg} \neq r}{}{\up}$, updating the registers in~$\up$, be fired at the $j$-th input~$(a_j,d_j)$
while reading~$v$; see Figure~\ref{fig:tree-ineq}.
We prove that the data word $u_1=(a_1,x_1)(a_2,x_1)\cdots(a_j,x_1)$ with $\data(u_1)=\{x_1\}$ 
guides~$\{\hat{\loc}\}\times \domain^k$ to a subset in which each configuration has some register with value~$x_1$:
$\post(\{\hat{\loc}\}\times \domain^k, u_1)\subseteq  \langle \locs, 1\rangle$.
This phenomenon is depicted in Figure~\ref{fig:tree-ineq-d1} and can be argued as follows.
Observe that $x_1=d_1$ is the first input datum; thus 
after inputting~$(a_1,x_1)$ the set of successors is a disjoint union of two branches:
\begin{itemize}
	\item either at least one register~$r$ has datum~$x_1$ after  the transition $\tto{\bigvee_{r\in \reg}= r}{a_1}$.
		All the following successors in this branch, on input~$(a_2,x_1)(a_3,x_1)\cdots(a_j,x_1)$, preserve the datum $x_{1}$ in the register $r$;

	\item	or none of the registers is assigned~$x_1$ after  the transition $\tto{\bigwedge_{r\in \reg}\neq r}{a_1}$. 
		By inputting ~$(a_2,x_1)(a_3,x_1)\cdots(a_j,x_1)$, all the  following successors in this branch, thus, take inequality-guarded transitions, 
		and would not update any registers, except for the last transition~$\ttto{ \bigwedge_{r\in \reg} \neq r}{}{\up}$ fired by~$(a_j,x_1)$. 

\end{itemize}
The above argument proves that $u_1$ with $\data(u_1)\subseteq \{x_1\}$ is such that 
$\post(\{\hat{\loc}\}\times \domain^k, u_1)\subseteq \langle \locs,1\rangle$. 
The base of induction holds. 

\begin{figure}[h]
\begin{minipage}{.45\textwidth}
 \begin{center}
\scalebox{.7}{\tikzset{
  treenode/.style = {shape=rectangle, rounded corners,
                     draw, align=center},
  root/.style     = {treenode},
  env/.style      = {treenode},
  dummy/.style    = {circle,draw}
}
\begin{tikzpicture}
  [
    grow                    = down,
    sibling distance        = 11em,
    level distance          = 5em,
  ]

  \node at (2,-2.5) {$\dots$}; 
 
  \node[root] (t){$\{\hat{\ell}\}\times \domain^k$} 
    child { node[env] (t1) {$\{\ell_{2}\}\times(\domain\setminus\{d_{1}\})^k$}
		child { node (t10) {if $d_1\neq d_2$ then} 
		edge from parent[->] node [midway,left] {$\bigvee_{r\in\reg} =r\;$} } 
    child { node (t11) {$\vdots$}
		child { node [env] (t111)  {$\{\ell_{j}\}\times (\domain\setminus\{d_{1}\times d_{2},\cdots,d_{j-1}\})^k $}
		child { node[env] (t1111) {$\{ \ell_{j+1}\}\times\{\nu\in \domain^k \; | \; \nu(r)= \left\{ \begin{array}{ll} d_{j}  & r\in\up \\ 	d\in \domain\setminus \{d_{1},\cdots, d_{j}\} & r\not\in\up \end{array}\right. \}$}
    edge from parent[->] node [right] {$\bigwedge_{r\in\reg} \neq r , \up\! \downarrow$} }
		edge from parent[->] node [ midway, left] {$\bigwedge_{r\in\reg} \neq r \;$} }
		edge from parent[->] node [ midway, left] {$ \bigwedge_{r\in\reg} \neq r \;$ } }
		edge from parent[->] node [ midway, left] {$\;\; \bigwedge_{r\in\reg} \neq r$} }
		child { node[env] (t0) {$\{\ell_{1}\}\times \domain^k \setminus(\domain\setminus\{d_{1}\})^k$} 
		edge from parent[->] node [midway,right] {$\bigvee_{r\in\reg} =r\;$} };
      
\end{tikzpicture}}\end{center}
\caption{Runs of~$\R$ over the data word~$(a_1,d_1)(a_2,d_2)\cdots(a_j,d_j)$.}
\label{fig:tree-ineq}
\end{minipage}
\begin{minipage}{.05\textwidth}\end{minipage}
\begin{minipage}{.5\textwidth}
 \begin{center}
\scalebox{.7}{\tikzset{
  treenode/.style = {shape=rectangle, rounded corners,
                     draw, align=center},
  root/.style     = {treenode},
  env/.style      = {treenode},
  dummy/.style    = {circle,draw}
}
\begin{tikzpicture}
  [
    grow                    = down,
    sibling distance        = 15em,
    level distance          = 5em,
  ]
		
  \node[root] (t){$\{\hat{\ell}\}\times \domain^k$} 
    child { node[env] (t1) {$\{\ell_{2}\}\times(\domain\setminus\{x_{1}\})^k$}
		child { node (t10) {\text{No successor!}} 
		edge from parent[->] node [midway,left] {$\bigvee_{r\in\reg} =r\;\;$} } 
    child { node (t11) {$\vdots$}
		child { node [env] (t111)  {$\{\ell_{j}\}\times (\domain\backslash\{x_{1}\})^k $}
			child { node[xshift=-2cm,yshift=2mm] (t110) {\text{No successor!}} 
		edge from parent[->] node [midway,left] {$\bigvee_{r\in\reg} =r\;\;$} } 
		child { node[env,xshift=-1.3cm] (t1111) {$\{ \ell_{j+1}\}\times\{\nu\in \domain^k \; | \; \nu(r)= \left\{ \begin{array}{ll} x_{1}  & r\in\up \\ 	d\in \domain\setminus \{x_{1}\} & r\not\in\up \end{array}\right. \}$}
    edge from parent[->] node [right] {$\; \; \bigwedge_{r\in\reg} \neq r , \up\!\downarrow$} }
		edge from parent[->] node [ midway, left] {$\bigwedge_{r\in\reg} \neq r\;\;$} }
		edge from parent[->] node [ midway, left] {$\bigwedge_{r\in\reg} \neq r\;\; $} }
		edge from parent[->] node [ midway, left] {$ \bigwedge_{r\in\reg} \neq r\;\;$} }
		    child { node[env] (t0) {$\{\ell_{1}\}\times \domain^k \setminus(\domain\setminus\{x_{1}\})^k$} 
		child { node  (t000) {
		\parbox[c]{4cm}{Note that all successors of this branch always preserve  the value $x_{1}$ on the register $r$ which satisfies the  guard $=r$ of the first transition.
		 }}
		edge from parent[dotted] node [midway,right] {};
		}
		edge from parent[->] node [midway,right] {$\bigvee_{r\in\reg} =r\;\;$} } ;
      
\end{tikzpicture}}\end{center}
 \caption{Runs of~$\R$ over the data word~$u_1=(a_1,x_1)(a_2,x_1)\cdots(a_j,x_1)$}
 \label{fig:tree-ineq-d1}
\end{minipage}
\end{figure}

\noindent {\bf Step of induction.} Assume that the induction hypothesis holds for $i-1$, namely,
there exists some word~$u_{i-1}$ with $\data(u_{i-1})\subseteq \{x_1,\cdots,x_{i-1}\}$ such that 
$\post(\{\hat{\loc}\}\times \domain^k, u_{i-1})\subseteq \langle \locs, i-1\rangle$.
To  construct $u_i$, we define the concept of a symbolic state:
we say $(\loc,\up,\val,j)$ is a symbolic state if
$\loc\in \locs$, the set~$\up \subseteq \reg$ of registers is such that $\abs{\up}\geq \min(j,k)$ and
$\val \in \{x_1,\cdots,x_{j}\}^k$ and $j\leq N$. 
The semantics of~$(\loc,\up,\val,j)$ is the following set:
$$\semantics{(\loc,\up,\val,j)}=\{\loc\}\times \{\val'\in \domain^k \mid \val' (r)=\val(r) \text{ if } r\in \up\}.$$
Denote by $\Gamma$ the set of  all such symbolic states~$(\loc,\up,\val,i-1)$. 
By definition, the set~$\Gamma$ is finite. 
Now we can construct~$u_i$ as follows.
Let $S_{0}=\post(\{\hat{\loc}\}\times \domain^k, u_{i-1})$ and $w_0=u_{i-1}$.
Recall that $S_0 \subseteq \langle \locs, i-1\rangle$
and observe that 
$S_0 \subseteq \bigcup_{q \in \Gamma}\semantics{q}$.
Start with $j=0$ and,  while $S_j\neq \emptyset$, 
pick a symbolic state~$q=(\loc,\up,\val,i-1)$ such that $\semantics{q}\cap S_{j}\neq \emptyset$
and construct a word~$u_q$ (as explained in the details below) such that 
\begin{itemize}
	\item $\data(u_q)=\{x_1,x_2,\cdots,x_i\}$, and
	\item $\post(\semantics{q},u_q)\subseteq \langle \locs, i\rangle$.
\end{itemize}
Let $S_{j+1}=\post(S_{j}\setminus \semantics{q},u_q)$ and $w_{j+1}=w_{j}\cdot u_q$.
Repeat the loop for $j+1$.
Observe that $u_{i}=w_{j^*}$, where $j^{*}\leq \abs{S_0}$ is such that $S_{j^*}=\emptyset$, satisfies  the induction statement.

\medskip

Below, given a symbolic state $q=(\loc,\up,\val,i-1)$, 
the aim is to construct the data word~$u_q$.  
Without loss of generality, we assume that~$\abs{\up}=i-1$; otherwise~$u_q=u_{i-1}$.
Let 
$$
\wait=\semantics{(\loc,\up,\val,i-1)} \cap \{\loc\}\times \{\valuation' \mid \valuation' (r)\in \domain \setminus \data(v) \text{ if } r \not\in \up\}
$$
be the set of all configurations in  the symbolic state~$q$, where all data stored in the registers~$r\not \in \up$ are not in $\data(v)$.
Similarly to the induction base, 
no matter what the register valuation in a configuration in~$\wait$ looks like, the unique run of~$\R$ on the synchronizing word~$v=(a_1,d_1)(a_2,d_2)\cdots (a_n,d_n)$ starting in that configuration  takes the same sequence 
of transitions. 
Since $\val\in \{x_0,\cdots,x_{i-1}\}^k$, 
after inputting successive data from~$\data(v)$, all  successors of configurations in $\wait$ are elements of a symbolic state.  
For all $0 \leq j \leq n$, let the symbolic state $q^j=(\loc^{j},\up^j,\val^j,N)$ be such that
$\semantics{q^0}=\semantics{q} \cap \wait$, and  $\post(\semantics{q^{j-1}},(a_{j},d_{j})) \subseteq \semantics{q^j}$ if $j\geq 1$.

\noindent  In the sequel, we argue that there exists some~$ 1\leq m \leq n$ 
such that, in the sequence of transitions from one symbolic state to another symbolic state over the 
prefix~$(a_1,d_1)(a_2,d_2)\cdots (a_m,d_m)$ of~$v$ (the first $m$~inputs), the following holds:
\begin{itemize}
	\item   on inputting~$(a_j,d_j)$ for all $1\leq j < m$,
	the transition~$\ttto{(\bigwedge_{r\in \Lambda_j} =r)\wedge (\bigwedge_{r\not \in \Lambda_j}\neq r)} {a_j}{\Gamma_j}$ 
	with $\Lambda_j,\Gamma_j \subseteq \up$ is taken
	from $q^{j-1}$ to~$q^{j}$.
	It implies that $\val^{j-1}(r)=d_{j}$ for all $r\in \Lambda_j$, and $\val^{j}(r)=d_{j}$ for all $r\in \Gamma_j$.

	\item and  on inputting~$(a_m,d_m)$, the 
	transition~$\ttto{(\bigwedge_{r\in \Lambda_m} =r)\wedge (\bigwedge_{r\not \in \Lambda_m} \neq r)} {a_m}{\Gamma_m}$, 
	that is taken
	from $q^{m-1}$ to~$q^{m}$,
	is
	such that~$\Lambda_m\subseteq \up^m$ whereas~$\Gamma_m\not \subseteq \up^m$.
\end{itemize}

\noindent
Now from the prefix~$(a_1,d_1)(a_2,d_2)\cdots (a_m,d_m)$ of~$v$, \ie, 
the first $m$~inputs, and from the set of data  $\{x_1,x_2,\cdots,x_i\}$, 
we construct the word~$u_q=(a_1,y_1)(a_2,y_2)\cdots(a_m,y_m)$ for $q=(\loc,\up,\val,i-1)$ as follows: 
for all $1\leq j \leq m$,

\begin{itemize}
\item if $\Lambda_j\neq \emptyset$, \ie, some register~$r\in \up$ already stores the datum~$d_{j}$,
then $y_{j}=d_j$.
\item if $\Lambda_j= \emptyset$, \ie, none of the registers~$r \in \up$ stores the datum~$d_j$, then
$y_{j}=d$ where $d\in \{x_1,x_2,\cdots,x_i\}\setminus \{\val^{j-1}(r)\mid r\in \up\}$.
The existence of such~$d$ is guaranteed since~$\abs{\up}=i-1$ and $\abs{\{x_1,x_2,\cdots,x_i\}}=i$.
Moreover, since the transitions $\ttto{(\bigwedge_{r\in \up}\neq r)} {a_j}{\Gamma_j}$ 
have inequality guards for all registers, then changing the datum from~$d_j$ to~$y_j$
would result only in taking the same transition. 
\end{itemize}

\noindent Observe that $\data(u_q)\subseteq\{x_1,\cdots,x_{i}\}$. 
As a result, all registers that are updated along
the runs of~$\R$ over~$u_q$ store some datum from~$\{x_1,\cdots,x_{i}\}$.
This  argument shows that $\post(\semantics{q},u_q)\subseteq \langle \locs, i\rangle$.
This concludes the step of induction, and completes the proof.
\end{proof}

After reading some word that shrinks the infinite set of configurations of DRAs to a finite set~$S$ of configurations,
we generalize the \emph{pairwise synchronization} technique~\cite{Volkov08} to finally  synchronize configurations in~$S$.
By this generalization, we achieve the following Lemma~\ref{lemmasyncD}, 
for which the detailed proof  can be found in Appendix~\ref{append_DRAs}.

\newcommand{\lemmasyncD}{
	For all DRAs for which there exist synchronizing data words, there exists a synchronizing data word $w$ such that $\abs{w}\leq 2\abs{\reg}+1$.
}

\begin{lemma} \label{lemmasyncD}
\lemmasyncD
\end{lemma}

Given a $1$-DRA~$\R$, the synchronization problem can be solved as follows: 
($1$) check that from each location~$\loc$ an update on the
single register is achieved by going through  inequality-guarded transitions, 
which can be done in~$\nlogspace$. 
Lemma~\ref{lemmafiniteD} ensures that feeding~$\R$ consecutively with  a single datum~$x\in \domain$ is sufficient 
for this phase and 
 the set of successors of~$\locs \times D$ would be a subset of~$\locs \times \{x\}$.
Next ($2$) pick an arbitrary set~$\{x,y,z\}$ of data including~$x$,  by Lemma~\ref{lemmasyncD}
and the pairwise synchronization technique, the  problem  reduces to the synchronization problem for
DFAs where  data  in registers and  input data extend  locations and the alphabet:
 $Q=\locs \times \{x,y,z\}$ and $\alphabet \times \{x,y,z\}$.
Since a $1$-DRA, where all transitions update the register and are  guarded with $\true$,  is equivalent to a DFA, we obtain the next theorem.

\newcommand{\theosDRA}{
The synchronization problem for $1$-DRAs is $\nlogspace$-complete.
}
\begin{theorem} \label{theosDRA}
\theosDRA
\end{theorem}

We provide a family of DRAs, for which a linear bound on the data efficiency of synchronizing data words, depending on the number of registers, is
necessary. 
This necessary and sufficient bound is crucial to establish membership of synchronizing DRAs in $\pspace$. 

\newcommand{\lemmafamilyDRA}{
	There is a family of single-letter DRAs $(\R_n)_{n\in \nat}$, 
	with $n=\abs{\reg}$ registers and $\mathcal{O}(n)$ locations,  such that  all  synchronizing data words
	have data efficiency~$\Omega(n)$.
}

\begin{lemma} \label{lemmafamilyDRA}
\lemmafamilyDRA
\end{lemma}
\begin{proof}
	The family of DRAs~$\R_n (n\in \nat)$  is defined over an infinite data domain~$\domain$.
The DRA~$\R_n$ has $n$~registers and a single letter~$a$.
The structure of~$\R_n$ is composed of two distinguished locations~$\init$ and~$\synch$ and two chains,
where each chain has~$n$ locations: $\loc_1,\loc_2,\cdots,\loc_n$ and $\loc'_1,\loc'_2,\cdots,\loc'_n$.
The DRA~$\R_3$ is shown in \figurename~\ref{fig:UpdateEfficienyRegisterAutomata}.
The only transition in~$\synch$ is a self-loop with update on all $n$~registers, 
thus $\R_n$ can only be synchronized in~$\synch$.
There are two transitions in~$\init$, each going to one of the chains: 
$$\init \ttto{=r_1}{a}{r_1} \loc_1 \text{~~~~~and~~~~~} \init \ttto{\neq r_1}{a}{r_1} \loc'_1.$$
Then, $\post(\{\init\}\times \domain^n, (a,x))=\{\loc_1,\loc'_1\}\times (\{x\}\times \domain^{n-1})$ for all~$x\in \domain$.

From~$\{\loc_1,\loc'_1\}\times (\{x\}\times \domain^{n-1})$, informally speaking,  in both chains the respective  $i$-th locations are 
simultaneously reached after 
inputting $i$~distinct data: 
for all $1 \leq i <n$,
in each ~$\loc_i$ and~$\loc'_i$ there are two transitions. 
One transition is a self-loop, with a satisfied equality guard on at least one of the updated registers $r_1,\dots, r_i$ so far. 
The other transition 
 goes to the next location $\loc_{i+1}$ in the chain, with an inequality guard on all updated registers $r_1,r_2,\cdots, r_i$ so far, and an update on the next register $r_{i+1}$. 

$$\loc_i \tto{\bigvee_{r\in \{r_1,\cdots,r_i\}} (=r_i)}{a~~~} \loc_i \text{~~~and~~~} \loc_i \ttto{\bigwedge_{r\in \{r_1,\cdots,r_i\}} (\neq r_i)}{a}{r_{i+1}} \loc_{i+1},$$ 
$$\loc'_i \tto{\bigvee_{r\in \{r_1,\cdots,r_i\}} (=r_i)}{a~~~} \loc'_i \text{~~~and~~~} \loc'_i \ttto{\bigwedge_{r\in \{r_1,\cdots,r_i\}} (\neq r_i)}{a}{r_{i+1}} \loc'_{i+1}.$$
At the last locations~$\loc_n$ and~$\loc'_n$ of the two chains, there is 
one transition with inequality guards on all registers leaving the chain to~$\synch$, and
there is one transition which is, again, a self-loop with an equality constraint for at least one of the registers.  
$$\loc_n \ttto{\bigwedge_{r\in \reg} (\neq r_i)}{a~~~}{\reg} \synch \text{~~~and~~~} \loc_n \tto{\text{{\sf else}}}{a}\loc_n  \quad \quad \quad \loc'_n \ttto{\bigwedge_{r\in \reg} (\neq r_i)}{a~~~}{\reg} \synch \text{~~~and~~~} \loc'_n \tto{\text{{\sf else}}}{a}\loc'_n.$$

By construction, we see that $n+1$ distinct data values must be read for reaching~$\synch$ from the infinite set~$\{\init\}\times \domain^n$.
Since $\R_n$ can only be synchronized in~$\synch$, all synchronizing data words must have data efficiency at least~$n+1 \in \Omega(n)$.

It remains to prove that $\R_n$ has indeed some synchronizing word. 
Let $\{x_1,x_2,\cdots,x_{n+1}\}$ be a set of  $n+1$ distinct data values and 
$w_{\synch}=(a,x_1)(a,x_2)\cdots (a,x_n)(a,x_{n+1})$.
For the configuration space $\locs=\{\init,\synch, \loc_1,\cdots,\loc_n,\loc'_1,\cdots,\loc'_n\}$, 
observe that $\post(\locs\times \domain^n,w_{\synch})=\{(\synch,x_{n+1})\}$ and $\abs{\data(w_{\synch})}=n+1$.
The proof is complete.
\end{proof}

\begin{theorem} \label{theokDRA}
\theokDRA
\end{theorem}
\begin{proof}(Sketch)
	The synchronization problem for $k$-DRA is in $\pspace$ using the following   $\conpspace$ algorithm: 
($1$) pick a set~$X=\{x_1,x_2,\cdots,x_{2k+1}\}$ of distinct data values. ($2$) guess some location~$\loc\in \locs$ and check if there is no
word~$w\in (\alphabet\times \{x_1,x_2,\cdots,x_{k}\})^*$ with length~$\abs{w}\leq 2^{k\abs{\locs}\abs{\Sigma}}$
such that along firing transitions that arer inequality-guarded on all $k$~registers,
some registers are not updated. If ($2$) is satisfied, then return ``no'' (meaning that there is no synchronizing data word for the input $k$-DRA). Otherwise, ($3$) guess two configurations~$q_1,q_2 \in \locs \times X^k$ such that 
there is no word~$w\in (\alphabet\times X)^*$ with length~$\abs{w}\leq 2^{(2k+1)\abs{\locs}\abs{\Sigma}}$
such that $\abs{\post(\{q_1,q_2\},w)}=1$. If ($3$) is satisfied, then the algorithm returns ``no''; otherwise return ``yes''.

For $\pspace$-hardness, we adapt an established reduction (see, e.g.,~\cite{DBLP:conf/fsttcs/0001JLMS14}) from the non-emptiness problem for $k$-DRA, see Appendix \ref{append_DRAs}. 
The result then follows by $\pspace$-completeness of the non-emptiness problem for $k$-DRA~\cite{DBLP:journals/tocl/DemriL09}.

\end{proof}

\section{Synchronizing data words for NRA\lowercase{s}}\label{synchNRAs}

In this section, we study the synchronization problems for NRAs.
We slightly update a result in~\cite{DBLP:conf/fsttcs/0001JLMS14} to present a general reduction from the  \emph{non-universality} problem to
the synchronization problem for NRAs. This reduction proves the  \emph{undecidability} result for the synchronization problem for $k$-NRAs, 
and $\ackermann$-hardness in $1$-NRAs. We then prove that for $1$-NRAs, the synchronization  and  non-universality problems
are indeed interreducible, which completes the picture by  $\ackermann$-completeness of the synchronization problem for~$1$-NRAs.
 
In the nondeterministic synchronization setting, we present two kinds of \emph{counting features}, which are useful for later constructions. 
For the first one, we define a family $(\R_{\counter(n)})_{n\in\nat}$ of $1$-NRAs with size only linear in $n$, where an input datum~$x\in \domain$ must be  read  $2^n$ times to achieve synchronization.

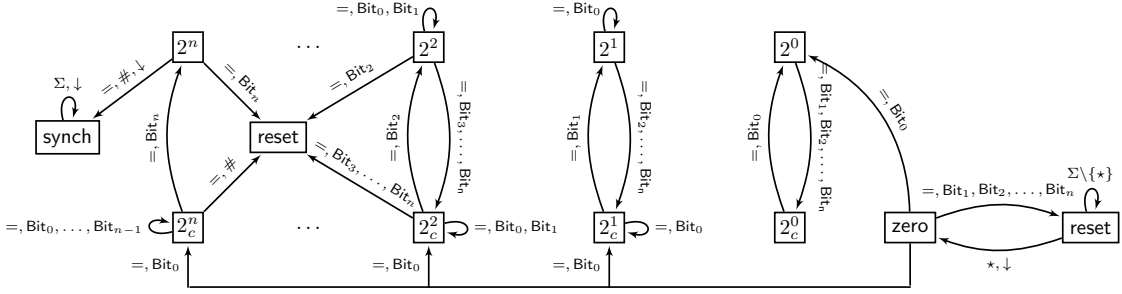
\begin{figure}[t]
\centering
  \scalebox{.8}{\begin{tikzpicture}[>=latex',shorten >=1pt,node distance=3.2cm,on grid,auto, thick,
		roundnode/.style={circle, draw=black,  thick, minimum size=6mm},
squarednode/.style={rectangle, draw=black,  thick, minimum size=5mm}]

\definecolor{myblue}{RGB}{80,80,160}
\definecolor{mygreen}{RGB}{80,160,80}

\node[squarednode] (k) at (0,3) {};
\node at (0,3) {$\loctwo{n}$};
\node[squarednode] (2) at (4,3) {};
\node at (4,3) {$\loctwo{2}$};
\node[squarednode] (1) at (7,3) {};
\node at (7,3) {$\loctwo{1}$};
\node[squarednode] (0) at (10,3) {};
\node at (10,3) {$\loctwo{0}$};

\node[squarednode] (ck) at (0,0) {};
\node  at (0,0) {$\locctwo{n}$};
\node[squarednode] (c2) at (4,0) {};
\node at (4,0) {$\locctwo{2}$};
\node[squarednode] (c1) at (7,0) {};
\node  at (7,0) {$\locctwo{1}$};
\node[squarednode] (c0) at (10,0) {};
\node  at (10,0) {$\locctwo{0}$};

\node[squarednode] (reset) at (15,0) {$\lreset$};
\node[squarednode] (reset2) at (1.5,1.5) {$\lreset$};
\node[squarednode] (synch) at (-2,1.5) {$\synch$};
\node[squarednode] (zero) at (12,0) {$\lzero$};

\path[-] (12.01,-.99) edge (-0.01,-.99);
\path[-] (zero) edge (12,-1);
\path[->] (7,-1) edge  node [left] {\scriptsize{$=,\Bit{0}$ }} (c1);
\path[->] (4,-1) edge  node [left] {\scriptsize{$=,\Bit{0}$ }} (c2);
\path[->] (0,-1) edge  node [left] {\scriptsize{$=,\Bit{0}$ }} (ck);

\node[draw=none] (Dots) at (2,0) {$\ldots$};
\node[draw=none] (Dots2) at (2,3) {$\ldots$};


\path[->] (synch) edge [loop above]  node {\scriptsize{$\Sigma, \downarrow$ }} ();

\path[->] (k) edge node [midway,sloped,above] {\scriptsize{$=,\#,\downarrow$ }}  (synch); 
\path[->] (ck) edge node [midway,sloped,above] {\scriptsize{$=,\#$ }}  (reset2);


\path[->] (reset) edge [loop above]  node {\scriptsize{$\Sigma\backslash\{\star\}$ }} ();


\path[->] (2) edge node [midway,sloped,above] {\scriptsize{$=,\Bit{2}$ }}  (reset2); 
\path[->] (k) edge node [midway,sloped,above] {\scriptsize{$=,\Bit{n}$ }}  (reset2); 
\path[->] (c2) edge node [midway,sloped,above] {\scriptsize{$=,\Bit{3},\dots,\Bit{n}$ }}  (reset2);


\draw [->] (reset) edge[bend left=20] node {\scriptsize{$\star, \downarrow$ }} (zero);
\draw [->] (zero) edge[bend left=20, sloped] node  {\scriptsize{$=,\Bit{1},\Bit{2},\dots,\Bit{n}$ }} (reset);

\path[->] (zero) edge[bend right=32] node [sloped, above] {\scriptsize{$=,\Bit{0}$ }}  (0); 
      
     
     
     


    
     \path[->] (1) edge [loop above]  node [left] {\scriptsize{$=,\Bit{0}$ }} ();
               
     \path[->] (2) edge [loop above]  node [left] {\scriptsize{$=,\Bit{0},\Bit{1}$ }} ();

     \path[->] (c1) edge [loop right]  node [right] {\scriptsize{$=,\Bit{0}$ }} ();
       
     \path[->] (c2) edge [loop right]  node [right] {\scriptsize{$=,\Bit{0},\Bit{1}$ }} ();

     \path[->] (ck) edge [loop left]  node [left] {\scriptsize{$=,\Bit{0},\dots,\Bit{n-1}$ }} ();
     

     
     \path[->] (c0) edge[bend left=20] node [midway,sloped,above] {\scriptsize{$=,\mathsf{Bit_{0}}$ }}  (0);

         
     \path[->] (0) edge[bend left=20] node [midway,sloped,above] {\scriptsize{$=,\mathsf{Bit_{1}},\mathsf{Bit_{2}},\dots,\mathsf{Bit_{n}}$ }}  (c0);

     
     \path[->] (c1) edge[bend left=20] node [midway,sloped,above] {\scriptsize{$=,\mathsf{Bit_{1}}$ }}  (1);

         
      	       \path[->] (1) edge[bend left=20] node [midway,sloped,above] {\scriptsize{$=,\mathsf{Bit_{2}},\dots,\mathsf{Bit_{n}}$ }}  (c1);

     
       \path[->] (c2) edge[bend left=20] node [midway,sloped,above] {\scriptsize{$=,\mathsf{Bit_{2}}$ }}  (2);

         
      \path[->] (2) edge[bend left=20] node [midway,sloped,above] {\scriptsize{$=,\mathsf{Bit_{3}},\dots,\mathsf{Bit_{n}}$ }}  (c2);

       
     
       \path[->] (ck) edge[bend left=20] node [midway,sloped,above] {\scriptsize{$=,\Bit{n}$ }}  (k); 
  \end{tikzpicture}}
\caption{
	A partial picture of the $1$-NRA $\R_{\counter(n)}$ (with $n\geq 3$) implementing a binary counter. 
	In order to avoid crossing edges in the figure, we use two copies of the same location $\lreset$.
	All locations have inequality-guarded self-loops for all letters in $\Sigma\backslash\{\star\}$. 
	All missing equality-guarded $\star$-transitions are directed to $\lzero$. For all $0\leq i <n$, missing equality-guarded $\#$-transitions from $\locctwo{i}$ are guided to $\synch$ with an update on the register. 
	All other non-depicted equality-guarded transitions are directed to $\lreset$, and inequality-guarded transitions are self-loops. 
}
\label{fig:countingNRA}
\end{figure}

\newcommand{\lemmacountingNRA}{
	There is a family of $1$-NRAs $(\R_{\counter(n)})_{n\in \nat}$  
with  $\mathcal{O}(n)$ locations,  
such that for all synchronizing data words~$w$,  
some datum~$d\in \data(w)$  appears in~$w$ at least $2^n$~times.
}

\begin{lemma} \label{lemmacountingNRA}
\lemmacountingNRA
\end{lemma}
\begin{proof}(Sketch)
	The $1$-NRA~$\R_{\counter(n)}$  shown in \figurename~\ref{fig:countingNRA} encodes a binary counter that  
	ensures that
	in every synchronizing data word~$w$ some  datum~$x\in \data(w)$  appears at least $2^n$ times. 
The location~$\synch$ has self-loops on all letters,
thus,  $\R_{\counter(n)}$ can only be synchronized in location~$\synch$. 
Generally speaking, the counting involves an \emph{initializing process} and several  \emph{incrementing processes}. 
The \emph{initializing process} is started by firing a $\star$-transition, which places a token, let us say: an $x$-token, into location $\lzero$. This sets the counter to $\inBinary{0}$. 
Note that firing $\star$-transitions is the only way to guide tokens out of $\lreset$; hence, whenever there is some token in $\lreset$, a new initializing process must be started. We use this to enforce a new initializing process whenever some transition is fired that is incorrect with respect to the incrementing process.
 
An \emph{incrementing process} can be set off by inputting the datum~$x$ via equality guards. 
The numbers $1\leq m\leq 2^n$ are represented by placing a copy of the $x$-token in the locations
corresponding to the binary representation of $m$. 
An $x$-token in location~$\loctwo{i}$ (in~$\locctwo{i}$, respectively) means that the $i$-th least 
significant  in the binary representation is set to~$\inBinary{1}$ (to~$\inBinary{0}$, respectively).
First, a $\Bit{0}$-transition places a copy of the $x$-token in each of  $\{\locctwo{n},\dots,\locctwo{2},\locctwo{1},\loctwo{0}\}$ to 
represent~$\inBinary{0\dots 001}$.
In each incrementation step the $x$-tokens are re-placed by firing specific $\Bit{i}$-transitions ($0\leq i \leq n$),
following the standard procedure of binary incrementation. 
At the end, when a copy of the $x$-token locates in each of $\{\loctwo{n},\locctwo{n-1},\dots,\locctwo{0}\}$ (representing $\inBinary{10\dots 0}$), 
the $\#$-transitions guide all of these tokens to location $\synch$ and finally synchronize~$\R_{\counter}$. 
We give a detailed explanation of the structure of~$\R_{\counter(n)}$ in Appendix \ref{append_NRAs}. 
\end{proof}

We present a second kind of counting features in RAs that explains the hardness of synchronizing NRAs, 
even with a single register.
In Lemma \ref{lemmatowerNRA}, 
we define a family of $1$-NRAs (with only $\mathcal{O}(n)$~locations), where  $\tower(n)$ distinct data must be read to gain synchronization. 
Recall from~\cite{DBLP:journals/corr/Schmitz13} that the function $\tower$ is at level three of the infinite \emph{Ackermann hierarchy} $(A_k)_{k\in\nat}$ of fast-growing  functions $A_i:\nat\to\nat$, inductively defined by 
$A_1(n)=2n$ and $A_{k+1}(n)=A_k^{n}(1)=\underbrace{A_k(\dots(A_k}_{n \, \mathrm{times}}(n))\dots)$. 
Hence, applying $\adouble\defequals A_1, \aexp\defequals A_2,$ and $\tower\defequals A_3$, respectively, on some natural number $n$ 
results in some number that is \emph{double}, \emph{exponential}, and \emph{tower}, respectively, in $n$. 
The function $A_\omega(n)=A_n(n)$ is a non-primitive recursive Ackermann-like function, defined by diagonalization.

\newcommand{\lemmatowerNRA}{
	There is a family of $1$-NRAs $(\R_{\tower(n)})_{n\in \nat}$  
with  $O(n)$ locations,  
such that  $\abs{\data(w)}\ge\tower(n)$ for all synchronizing data words~$w$.
}

\begin{lemma} \label{lemmatowerNRA}
\lemmatowerNRA
\end{lemma}
\begin{proof}
	The domain of the family of $1$-NRAs~$(\R_{\tower(n)})_{n\in \nat}$  is the natural numbers~$\nat$.
	The alphabet of $\R_{\tower(n)}$ is~$\alphabet=\{\#,\star, \arepeat,\adouble,\aexp,\atower\}$.
The structure of~$\R_{\tower(n)}$ is composed of 
$n$ locations $\Data_1, \Data_{1,2},\cdots ,\Data_{1,2,\cdots,n}$
and $6$ more locations~$\lreset,\synch,\store,\rep,\waitdoub,\waitexp$. The general structure of~$\R_{\tower(n)}$ is partially depicted in \figurename~\ref{fig:countingNRA-tower}. 
The NRA~$\R_{\tower(n)}$ is such that $\abs{\data(w)}\ge \tower(n)$ for all synchronizing data words~$w$.

All transitions in~$\synch$ are self-loops with an update on  the register $\synch \ttto{}{\alphabet}{r} \synch$;
thus,  $\R_{\tower(n)}$ can only be synchronized in~$\synch$. 
Moreover, $\synch$ is only accessible from $\store$ by a $\#$-transition. 
Assuming~$w$ is one of the shortest synchronizing words, we see that
$\post(\locs\times \domain,w)=\{(\synch,x)\}$, where $w$ ends with $(\#,x)$.

From all locations~$\ell\in \locs \setminus\{\synch\}$, we have
$\ell  \ttto{}{\star}{r} \Data_1$;
we say that $\star$-transitions \emph{reset}~$\R_{\tower(n)}$.
Moreover, the only outgoing transition in location~$\lreset$ is the $\star$-transition.
Thus, a \emph{reset} must occur in order to synchronize $\R_{\tower(n)}$. 
After this forced reset, say on reading~$(\star,1)$, the set of reached configurations is~$\{(\Data_1,1),(\synch,1)\}$. 
Since resetting is inefficient, we try to avoid it;
we call all transitions leading to~$\lreset$ \emph{inefficient}.

\noindent For all locations $\Data_{1,\cdots,i}$ with $1 \leq i < n$,
we define the two transitions 
$$\Data_{1,\cdots,i} \tto{\neq r}{\arepeat} \Data_{1,\cdots,i+1} \quad\text{~~~~~and~~~~~}\quad\Data_{1,\cdots,i} \ttto{\neq r}{\arepeat}{r} \Data_{1,\cdots,i+1}.$$
All other transitions in $\Data_{1,\cdots,i}$ are inefficient and directed to $\lreset$.
Below, we  rename $\Data_{1,2,\cdots,n}$ to $\waittow$.
We partially depict the transitions from $\waittow$, $\waitexp$, $\waitdoub$, $\rep$ and $\store$ in \figurename~\ref{fig:countingNRA-tower}.
All  transitions  are inefficient, except
\begin{itemize}
\item $\waittow \tto{= r}{\atower} \waitexp$, $\waittow \tto{\neq r}{\atower}\waittow$, and $\waittow \tto{\sigma}{}\waittow$ for all $\sigma\in\{\adouble,\aexp,\arepeat\}$.
	\item $\waitexp \tto{= r}{\aexp}\waitdoub$, $\waitexp \tto{}{\adouble}\waitexp $ and $\waitexp \tto{}{\arepeat}\waitexp$.
	\item $\waitdoub \tto{= r}{\adouble} \rep$, $\waitdoub \tto{\neq r}{\adouble}\waitdoub$ and $\waitdoub \tto{\neq r}{\arepeat}\waitdoub$,
	\item $\rep \tto{\neq r}{\arepeat} \store$ and $\rep \ttto{\neq r}{\arepeat}{r} \store$,
	\item $\store \tto{\atower}{} \waitexp$, $\store \tto{\aexp}{} \waitdoub$, $\store \tto{\neq r}{\adouble}\store$  and $\store \tto{\neq r}{\arepeat}\store$, and 
	\item $\store \ttto{}{\#}{r} \synch$. 
\end{itemize}
\begin{figure}
\begin{center}
	\scalebox{.85}{
		\begin{tikzpicture}[>=latex',shorten >=1pt,node distance=1.9cm,on grid,auto, thick,
roundnode/.style={circle, draw=black, thick, minimum size=2mm},
squarednode/.style={rectangle, draw=black, thick, minimum size=3mm}]

\definecolor{myblue}{RGB}{80,80,160}
\definecolor{mygreen}{RGB}{80,160,80}

\definecolor{myblue}{RGB}{80,80,160}
\definecolor{mygreen}{RGB}{80,160,80}

\node[squarednode] (reset) at (0,0) {$\lreset$};
\node[squarednode] (data1) at (3,0) {$\Data_{1}$};
\node[squarednode] (data2)  at (6,0) {$\Data_{1,2}$};
\node[draw=none] (Dots) at (9,0) {$\ldots$};
\node [squarednode] (wt)  [draw, text width=2cm] at (12,0) {$\;\; \waittow$ \\ $(\Data_{1,2,\dots,n})$};
\node[squarednode] (wd)  at (12,-2.4) {$\waitexp$};
\node[squarednode] (wr)  at (9,-2.4) {$\waitdoub$}; 
\node[squarednode] (lr)  at (6,-2.4) {$\rep$};
\node[squarednode] (st)  at (3,-2.4) {$\store$};
\node[squarednode] (synch)  at (0,-2.4) {$\synch$};

\path[-] (3.1,-2.6) edge (3.1,-3.3);
\path[-] (3.1,-3.3) edge (9,-3.3);
\path[->] (9,-3.3) edge  node [left] {\scriptsize{$\aexp$ }} (wr);

\path[-] (2.9,-2.6) edge (2.9,-3.6);
\path[-] (2.9,-3.6) edge (12,-3.6);
\path[->] (12,-3.6) edge  node [left] {\scriptsize{$\atower$ }} (wd);

   \path[->] (reset) edge node {\scriptsize{$\star, \downarrow$ }} (data1);
   
   \path[->] (data1) edge [bend right] node [midway,below] {\scriptsize{$\neq,\arepeat$ }} (data2);   
  \path[->] (data1) edge [bend left] node [midway,above] {\scriptsize{$\neq,\arepeat, \downarrow$ }} (data2);

  \path[->] (data2) edge [bend right] node [midway, below] {\scriptsize{$\neq,\arepeat$ }} (Dots);   
  \path[->] (data2) edge [bend left] node [midway,above] {\scriptsize{$\neq,\arepeat,\downarrow$ }} (Dots);

  \path[->] (Dots) edge [bend left=-26] node [midway, below] {\scriptsize{$\neq,\arepeat$ }} (10.9,-.2);   
  \path[->] (Dots) edge [bend right=-26] node [midway,above] {\scriptsize{$\neq,\arepeat,\downarrow$ }} (10.9,.2);

  
  \path [->] (wt) edge [loop above,looseness=5]  node [above] {\scriptsize{$\aexp,\arepeat,\adouble$}} (wt);
  
  \path [->] (wt) edge [loop right,looseness=3]  node [right] {\scriptsize{$\neq, \atower$}} (wt);


  \path[->] (wt) edge node [right] {\scriptsize{$=,\atower$ }} (wd);

  \path[->] (wd) edge node [midway,sloped,above] {\scriptsize{$=,\aexp$ }} (wr);

  \path[->] (wr) edge node [midway,sloped,below] {\scriptsize{$=,\adouble$ }} (lr);

  \path[->] (lr) edge [bend left] node [midway,sloped, below] {\scriptsize{$\neq,\arepeat$ }} (st);   
  \path[->] (lr) edge [bend right] node [midway,sloped,above] {\scriptsize{$\neq,\arepeat,\downarrow$ }} (st);


  \path[->] (st) edge node [below] {\scriptsize{$\#,\downarrow$ }} (synch);

\path [->] (synch) edge  [loop above] node [above] {\scriptsize{$\Sigma,\downarrow$\,}}  (synch);

\path [->] (wd) edge [out=-12,in=4,looseness=6] node [right=0.05cm,yshift=1mm] {\scriptsize{$\arepeat,\adouble$ }} node [right=0.05cm,yshift=-1mm] {\scriptsize{$\neq,\aexp$ }} (wd);



\path [->] (st) edge  [loop above] node [above] {\scriptsize{$\neq,\arepeat\ \ $ }} node [above,yshift=2mm] {\scriptsize{$\neq,\adouble$ }}   (st);


\path [->] (wr) edge  [loop above] node [above] {\scriptsize{$\neq,\arepeat\ \ $ }} node [above,yshift=2mm] {\scriptsize{$\neq,\adouble$ }}   (wr);



 \end{tikzpicture}	
	}
	\end{center}
\caption{A partial illustration of the $1$-NRA  $\R_{\tower(n)}$ for $n\ge 3$. All $\star$-transitions are guided to $\Data_1$ with an update on the register. All other missing non-depicted transitions are directed to $\lreset$. }
\label{fig:countingNRA-tower}
\end{figure}
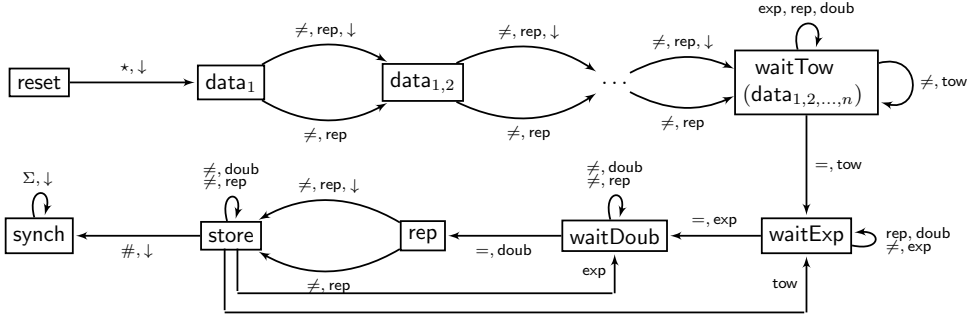
We remark that $\store \ttto{}{\#}{r} \synch$ is the only $\#$-transition that is not inefficient. This implies that for efficiently synchronizing $\R_{\tower(n)}$, one needs to re-move all produced tokens to $\store$ before firing a $\#$-transition.  The main issue in re-moving produced tokens, however, is  that some inequality-guarded transitions are unavoidable,  
and these transitions  may \emph{replicate}
the  tokens. For example, if one token  is in $\Data_1$,
 firing two transitions $\Data_1 \tto{\neq r}{ \arepeat} \Data_{1,2}$ and
$\Data_1 \ttto{\neq r}{\arepeat}{r} \Data_{1,2}$
replicates one token to two tokens in $\Data_{1,2}$.
Using this, 
one can implement \emph{doubling}, \emph{exponentialization}, and \emph{towering} of distinct tokens, as explained in the following.

\emph{Doubling:} Assume that there are $n$ distinct tokens $\{1,2,\dots,n\}$ in $\waitdoub$. 
Then the only efficient transition
is 
$\waitdoub\tto{= r}{\adouble} \waitrep$. 
In particular, all $\{\#,\aexp,\atower\}$-transitions activate a reset. As a result, as long as
some token is in $\waitdoub$,   $\{\#,\aexp,\atower\}$-transitions should be avoided for the sake of efficiency.  
This implies that for all $1\leq i\leq n$, the  $i$-token in $\waitdoub$ can 
leave the location only individually on the input~$(\adouble,i)$. 
Now, inputting $(\adouble,i)$ moves the $i$-token to $\waitrep$.  
Here the $i$-token must immediately move on to $\store$ via the inequality-guarded $\arepeat$-transitions, which will replicate the $i$-token into two tokens. 
Note that we must fire $\arepeat$-transitions with some ``fresh'' datum $j$ such that $j\not\in\{1,\dots,n\}$, otherwise a reset is evoked. (For simplicity, we  use $j=i+n$ by convention.)
It can now be easily seen that the only efficient way to guide all $n$ tokens out of $\waitdoub$ is by inputting the data word 
$$w_{\adouble(n)}=(\adouble,1)(\arepeat,n+1)(\adouble,2)(\arepeat,n+2)\dots(\adouble,n)(\arepeat,2n),$$
which puts $2n$ distinct tokens into $\store$.

\emph{Exponentialization:}
Assume there are $n$ distinct tokens $\{1,2,\dots,n\}$ in $\waitexp$. 
The only efficient transition
is $\waitexp\tto{= r}{\aexp} \waitdoub$. 
In particular, all $\{\#,\atower\}$-transitions activate a reset, and should be avoided as long as
some token is in $\waitexp$. 
This implies that for all $1\leq i\leq n$, the  $i$-token in $\waitexp$ can 
leave the location only individually on the input~$(\aexp,i)$. 
Now, inputting $(\aexp,1)$ moves the $1$-token to $\waitdoub$.  
From above we know that the only efficient way for guiding a single token in $\waitdoub$ towards synchronization is by inputting the data word $w_{\adouble(1)}$, resulting in two distinct tokens in $\store$:  $1$ and $2$.  
We can now proceed to remove the $2$-token from $\waitexp$ by inputting $(\aexp,2)$. Note that this also guides the  $\{1,2\}$-tokens residing in $\store$ to $\waitdoub$. 
Again, for efficient synchronization, we must input the data word $w_{\adouble(2)}$, which results in four distinct tokens $\{1,2,3,4\}$ in $\store$. 
It is now easy to see  
that the only efficient way to guide all $n$ tokens out of $\waitexp$ is by inputting the data word 
$$w_{\aexp(n)}=(\aexp,1) \, \cdot \, w_{\adouble(1)} \, \cdot \,(\aexp,2) \, \cdot \,w_{\adouble(2)} \, \cdot \,(\aexp,3) \, \cdot \,w_{\adouble(4)} \, \cdot \,\dots \, \cdot \,(\aexp,n) \, \cdot \,w_{\adouble(2^{n-1})},$$
which puts $2^n$ distinct tokens into $\store$.

\emph{Towering:}
Assume there are $n$ distinct tokens $\{1,2,\dots,n\}$ in $\waittow$.  
The only efficient transition
is $\waitexp\tto{= r}{\atower} \waitexp$. 
In particular, firing $\#$-transitions activates a reset, and should be avoided as long as
some token is in $\waittow$. 
This implies that for all $1\leq i\leq n$, the  $i$-token in $\waittow$ can 
leave the location only individually on the input~$(\atower,i)$. 
Now, inputting $(\aexp,1)$ moves the $1$-token to $\waitexp$. 
From above we know that the only efficient way for guiding a single token in $\waittow$ towards synchronization is by inputting the data word $w_{\aexp(1)}$, resulting in two distinct tokens in $\store$:  $1$ and $2$.  
We can now proceed to remove the $2$-token from $\waittow$ by inputting $(\atower,2)$. Note that this also guides the  $\{1,2\}$-tokens residing in $\store$ to $\waitexp$. 
Again, for efficient synchronization, we must input the data word $w_{\aexp(2)}$, which results in four distinct tokens $\{1,2,3,4\}$ in $\store$. 
It is now easy to see  
that the only efficient way to guide all $n$ tokens out of $\waittow$ is by inputting the data word 
$$w_{\atower(n)}=(\atower,1) \, \cdot \,w_{\aexp(1)} \, \cdot \,(\atower,2) \, \cdot \,w_{\aexp(2)} \, \cdot \,(\atower,3) \, \cdot \,w_{\aexp(4)} \, \cdot \,\dots \, \cdot \,(\atower,n) \, \cdot \,w_{\aexp(\tower(n-1))},$$
which puts $\tower(n)$ distinct tokens into $\store$.

Now, 
after the (forced) initial reset by firing $\star$-transitions, 
it is easy to see that the only data word that advances in synchronizing is $(\arepeat,2)(\arepeat,3)\cdots (\arepeat,n)$.
It replicates the $1$-token to
$n$ distinct tokens $1, 2, \cdots, n$, which are placed into~$\waittow$.  
From above we know that the only efficient way to guide all $n$ tokens out of $\waittow$ is by inputting $w_{\atower(n)}$, which places $\tower(n)$ distinct tokens into $\store$. 
We can now fire $\#$-transitions to synchronize~$\R_{\tower(n)}$ without evoking a reset, but note that due to the equality guard at the $\#$-transition from $\store$ to $\synch$, each of the $\tower(n)$ distinct tokens in $\store$ can move to $\synch$ only individually. 
This implies $\abs{\data(w)}\ge \tower(n)$ for all synchronizing words~$w$. 
\end{proof}

We can now use similar ideas as in Lemma~\ref{lemmatowerNRA} for defining a family of $1$-NRAs $\R_{A_n(m)}$ ($n,m\in \nat$)  such that all
synchronizing data words of $\R_{A_n(m)}$ have data efficiency at least~$A_n(m)$, where $A_n$ is at level $n$ of the Ackermann hierarchy. 
This provides a good intuition that the synchronization problem for NRAs must be $\ackermann$-hard, even 
if the NRA has a single register. 
In the following, we prove that the synchronization problem and the non-universality problem for NRAs are interreducible.

Let us first define the non-universality problem for RAs. 
To define the language of a given NRA~$\R$, we equip it with 
an initial location~$\loc_\is$ and a set $\locs_\acc$ of accepting locations,
where, without loss of generality, we assume that all outgoing transitions from~$\loc_\is$ update all registers.
The language~$L(\R)$ is the set of all data words~$w\in (\alphabet\times\domain)^*$, for which  
there is a run from~$(\loc_\is,\valuation_\is)$ to $(\loc_\acc,\valuation_\acc)$
such that $\loc_\acc \in \locs_\acc$ and $\valuation_\is,\valuation_f\in \domain^{\abs{\reg}}$.
The non-universality problem asks, given an RA, whether there exists some data word $w$ over $\alphabet$ such that  $w\not\in L(\R)$. 
We adopt an established reduction in~\cite{DBLP:conf/fsttcs/0001JLMS14} to provide the following Lemma.

\newcommand{\lemmaNonunivToSynch}{
The non-universality problem is reducible \mbox{to the synchronization problem for NRAs.}}
\begin{lemma}
	\label{lemma_nonuniv_to_synch}
	\lemmaNonunivToSynch
\end{lemma}
The detailed proof can be found in Appendix \ref{append_NRAs}. As an immediate result of Lemma~\ref{lemma_nonuniv_to_synch} and the undecidability of 
the non-universality problem for NRAs~(Theorems 2.7 and 5.4 in~\cite{DBLP:journals/tocl/DemriL09}), we obtain the following theorem.	
\begin{theorem} \label{theore-mundec-NRA}
	The synchronization problem for NRAs is undecidable.
\end{theorem}
Next, we present a reduction showing that, for~$1$-NRAs, the synchronization problem is reducible to the non-universality problem,
providing the tight complexity bounds for the synchronizing problem.

\newcommand{\lemmaReductionSynchUniv}{
The synchronization problem is reducible to \mbox{the non-universality problem for~$1$-NRAs.}}
\begin{lemma}
	\label{lemma_reduction_synch_univ}
	\lemmaReductionSynchUniv	
\end{lemma}
\begin{proof}
	We establish a reduction from the synchronization problem to the non-universality problem for~$1$-NRAs as follows.
	Given a $1$-NRA~$\R =\tuple{\locs,\reg,\Sigma,T}$, we construct a  $1$-NRA~$\R_{\mathsf{comp}}$ equipped with an initial location
	and a set of accepting locations such that $\R$ has some synchronizing word if, and only if, there exists some data word that is not in $L(\R_{\mathsf{comp}})$.

	First, we see that an analogue of Lemma~\ref{lemmafiniteD} holds for $1$-NRAs: 
for all $1$-NRAs with some synchronizing data word, 
there exists some word~$w$ with data efficiency~$1$
such that $\post(\locs\times \domain, w)\subseteq \locs \times \data(w)$. 
For all locations~$\loc \in \locs$, such a data word must  update the register by firing an inequality-guarded 
transition that is reached only via inequality-guarded transitions;
this can be checked  in~$\nlogspace$.
Given~$\R$, we assume that such a data word~$w$ always exists;
otherwise, we define $\R_{\mathsf{comp}}$ to be a $1$-NRA with a single (initial and accepting) location equipped with self-loops for all letters, 
so that  $L(\R_{\mathsf{comp}})=(\alphabet\times\domain)^*$.    
Given $\data(w)=\{x\}$, we say that $\R$ has some synchronizing word~$v$ if  
$\post(\locs\times \{x\}, v)$ is a singleton. 

\bigskip

Second, we define a data language $\lang$ such that data words in this 
language are encodings of the synchronizing process.
Let $\locs=\{\loc_1,\loc_2,\cdots,\loc_n\}$ be the set of locations and~$x,y$ two distinct data.
Informally, each data word in $\lang$ starts with the 
\begin{itemize}
	\item \emph{initial block}: a delimiter~$(\star,y)$, the sequence~$(\loc_1,x),(\loc_2,x),\cdots,(\loc_n,x)$ and 
		an input~$(a,d)\in \alphabet\times \domain$ as the beginning of a synchronizing word. The initial block is followed by several 
	\item \emph{normal blocks}: the delimiter~$(\star,y)$, the set of successor configurations reached from the configurations and the input of the previous block,
		and the next input $(a',d')$ of the synchronizing data word. The data word finally ends with the 
	\item \emph{final block}: the delimiter~$(\star,y)$, a single successor  configuration 
		reached from the configurations and the input of the previous block, and the delimiter~$(\star,y)$. 
\end{itemize}

\medskip
Formally, the  language $\lang$ is defined over the alphabet $\alphabet_{\lang}=\alphabet\cup \locs\cup\{\star\}$
where $\star \not \in\alphabet\cup\locs$. 
It contains all data words~$u$ that 
satisfy the following \emph{membership conditions}: 

\begin{enumerate}
\item The data words~$u$ starts with $(\star,y)(\loc_1,x),(\loc_2,x),\cdots,(\loc_n,x)$ for some $x,y\in \domain$ with $y\neq x$; 
	this condition guarantees  the correctness of the encoding for the initial block.

\item Let $\proj(u)$ be the projection of~$u$ into  $\alphabet_{\lang}$ (\ie, omitting the data values).
   Then there exists some~$\loc_{\synch}\in\locs$ where $\proj(u)\in (\star \,\locs^{+} \alphabet)^+\star \,\loc_{\synch}\, \star$.
   This condition guarantees  the right form of data words to be encodings of  synchronizing processes.

 \end{enumerate}

 The next two conditions  guarantee the uniqueness of the delimiter:
	\begin{enumerate}
		\addtocounter{enumi}{2} 
	\item The letter $\star$ in~$u$ occurs only with datum $y$. 
	\item No other letter in~$u$ occurs with datum~$y$.
\end{enumerate}

The next three conditions guarantee that all the successors that can be reached from configurations and inputs in each block are correctly 
inserted in the next block.
For all $(\loc,x)\in \locs\times \domain$  and $(a,d)\in \alphabet\times \domain$  in the same block,
	\begin{enumerate}
		\addtocounter{enumi}{4} 
	\item  if $x=d$ and there exists a transition $\loc\tto{= r}{a}\loc'$ (with or without update), then 
	  $(\loc',x)$   must be in the next block. 
	\item if $x\neq d$ and there exists a transition $\loc\tto{\neq r}{a}\loc'$, then 
	  $(\loc',x)$  must be in the next block. 
	\item if $x\neq d$ and there exists a transition $\loc\ttto{\neq r}{a}{r}\loc'$ then 
	  $(\loc',d)$ must be  in the next block. 
\end{enumerate}

By construction, the NRA $\R$ has some synchronizing data word 
if, and only if, $\lang \neq\emptyset$.
Below, we construct a $1$-NRA $\R_{\mathsf{comp}}$ that accepts the complement of $\lang$. 
Then, the NRA $\R$ has some synchronizing data word 
if, and only if, there exists some data word that is not in $L(\R_{\mathsf{comp}})$.

The $1$-NRA $\R_{\mathsf{comp}}$ is the union of several $1$-NRAs that are in the family of $1$-NRAs $\Rfamily_1,\Rfamily_2,\cdots, \Rfamily_7$, where an $1$-NRA is in the family $\Rfamily_i$ if it violates the  $i$-th 
condition among the membership conditions in~$\lang$.
 
\begin{enumerate}
\item Family~$\Rfamily_1$: we add a $1$-NRA that accepts data words not starting with~$(\star,y)(\loc_1,x),\cdots,(\loc_n,x)$.
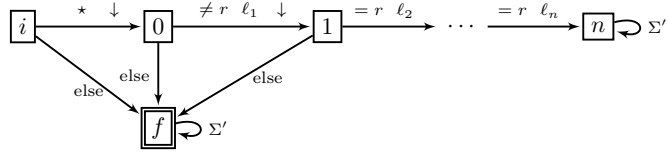
\begin{center}
\scalebox{.9}{
\begin{tikzpicture}[>=latex',shorten >=1pt,node distance=1.5cm,on grid,auto, thick,
roundnode/.style={circle, draw=black,  thick, minimum size=2mm},
squarednode/.style={rectangle, draw=black, thick, minimum size=3mm}]
 
\node[squarednode] (1)   {$i$}; 
\node[squarednode] (2)  [right =2cm of 1] {$0$};
\node[squarednode] (3)  [right=2.5cm of 2] {$1$};
\node[squarednode,draw=none] (4)  [right =2cm of 3] {$\cdots$};
\node[squarednode] (n)  [right =2cm of 4] {$n$};
\node[squarednode,accepting] (f)  [below =1.5cm of 2] {$f$};

\path [->] (1) edge node [midway,above] {\scriptsize{$~~\star~~\downarrow$}} (2);
\path [->] (1) edge node [midway,below] {\scriptsize{else }} (f);
\path [->] (2) edge node [sloped,above] {\scriptsize{$\neq r~~\loc_1~~\downarrow$}} (3);
\path [->] (2) edge node [midway,left] {\scriptsize{else }} (f);
\path [->] (3) edge node [sloped,above] {\scriptsize{$= r~~\loc_2~~$}} (4);
\path [->] (3) edge node [midway,right] {\scriptsize{else }} (f);
\path [->] (4) edge node [sloped,above] {\scriptsize{$= r~~\loc_n~~$}} (n);
\path [->] (n) edge [loop right] node [right] {\scriptsize{$\alphabet'$ }} (n);
\path [->] (f) edge [loop right] node [right] {\scriptsize{$\alphabet'$ }} (f);

\end{tikzpicture}
}
\end{center}

\item Family~$\Rfamily_2$: we add a DFA  that accepts data words~$u$ such that $\proj(u)$ is 
not in the regular language~$(\star \,\locs^{+} \alphabet)^+\star \,\loc_{\synch}\, \star$.

\item Family~$\Rfamily_3$:  we add a~$1$-NRA that accepts data words in which 
	 two delimiters~$\star$ have different data. 
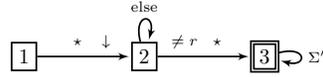
\begin{center}
\scalebox{.8}{
\begin{tikzpicture}[>=latex',shorten >=1pt,node distance=1.5cm,on grid,auto, thick,
roundnode/.style={circle, draw=black, thick, minimum size=2mm},
squarednode/.style={rectangle, draw=black, thick, minimum size=3mm}]
 
\node[squarednode] (1)   {$1$}; 
\node[squarednode] (2)  [right =2cm of 1] {$2$};
\node[squarednode,accepting] (3)  [right=2cm of 2] {$3$};

\path [->] (1) edge node [midway,above] {\scriptsize{$~~\star~~\downarrow$ }} (2);
\path [->] (2) edge [loop above] node [above] {\scriptsize{else}} ();
\path [->] (2) edge node [sloped,above] {\scriptsize{$\neq r~~\star~~$}} (3);
\path [->] (3) edge [loop right] node [right] {\scriptsize{$\alphabet'$ }} (3);
\end{tikzpicture}
}
\end{center}

\item Family~$\Rfamily_4$:  we add a~$1$-NRA that accepts data words in which 
the datum  of first $\star$ is not used only by occurrences of~$\star$. 
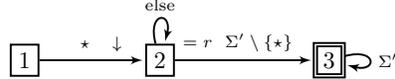
\begin{center}
\scalebox{.9}{
\begin{tikzpicture}[>=latex',shorten >=1pt,node distance=1.5cm,on grid,auto, thick,
roundnode/.style={circle, draw=black, thick, minimum size=2mm},
squarednode/.style={rectangle, draw=black, thick, minimum size=3mm}]
 
\node[squarednode] (1)   {$1$}; 
\node[squarednode] (2)  [right =2cm of 1] {$2$};
\node[squarednode,accepting] (3)  [right=2.5cm of 2] {$3$};

\path [->] (1) edge node [midway,above] {\scriptsize{$~~\star~~\downarrow$ }} (2);
\path [->] (2) edge [loop above] node [above] {\scriptsize{else}} ();
\path [->] (2) edge node [sloped,above] {\scriptsize{$=r~~\alphabet'\setminus\{\star\}~~$}} (3);
\path [->] (3) edge [loop right] node [right] {\scriptsize{$\alphabet'$ }} (3);
\end{tikzpicture}
}
\end{center}

\item Family~$\Rfamily_5$: for all transitions $\loc\tto{= r}{a}\loc'$, we add a~$1$-NRA that
only accepts data words such that  one block contains some $(\loc,x)$
and $(a,d)$ with $x= d$ where the next block does not have~$(\loc',x)$.
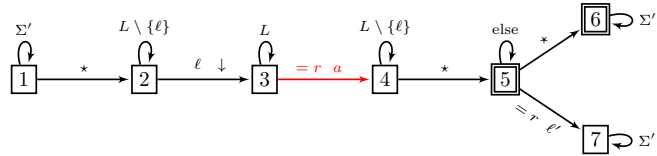
\begin{center}
\scalebox{.8}{
\begin{tikzpicture}[>=latex',shorten >=1pt,node distance=1.5cm,on grid,auto, thick,
roundnode/.style={circle, draw=black, thick, minimum size=2mm},
squarednode/.style={rectangle, draw=black, thick, minimum size=3mm}]
 
\node[squarednode] (1)   {$1$}; 
\node[squarednode] (2)  [right =2cm of 1] {$2$};
\node[squarednode] (3)  [right=2cm of 2] {$3$};
\node[squarednode] (4)  [right =2cm of 3] {$4$};
\node[squarednode,accepting] (5)  [right =2cm of 4] {$5$};
\node[squarednode,accepting] (6)  [above right=1cm and 1.5cm  of 5] {$6$};
\node[squarednode] (7)  [below =2cm of 6] {$7$};

\path [->] (1) edge [loop above] node [above] {\scriptsize{$\alphabet'$ }} ();
\path [->] (1) edge node [midway,above] {\scriptsize{$\star$ }} (2);
\path [->] (2) edge [loop above] node [above] {\scriptsize{$\locs\setminus\{\loc\}$ }} ();
\path [->] (2) edge node [sloped,above] {\scriptsize{$~~\loc~~\downarrow$}} (3);
\path [->] (3) edge [loop above] node [above] {\scriptsize{$\locs$ }} ();
\path [->,red] (3) edge node [sloped,above,red] {\scriptsize{$= r~~a~~$}} (4);
\path [->] (4) edge node [sloped,above] {\scriptsize{$\star$ }} (5);
\path [->] (4) edge [loop above] node [above] {\scriptsize{$\locs\setminus\{\loc\}$ }} ();
\path [->] (5) edge [loop above] node [above]  {\scriptsize{else}} ();
\path [->] (5) edge node [sloped,above] {\scriptsize{$\star$ }} (6);
\path [->] (5) edge node [sloped,below] {\scriptsize{$=r~~\loc'~~$}} (7);
\path [->] (6) edge [loop right] node [right] {\scriptsize{$\alphabet'$ }} (6);
\path [->] (7) edge [loop right] node [right] {\scriptsize{$\alphabet'$ }} (7);
\end{tikzpicture}
}
\end{center}

\item Family~$\Rfamily_6$: for all transitions $\loc\tto{\neq r}{a}\loc'$, we add a~$1$-NRA that
	only accepts data words such that  one block contains some $(\loc,x)$
	and $(a,d)$ with $x\neq d$ where the next block does not have~$(\loc',x)$.
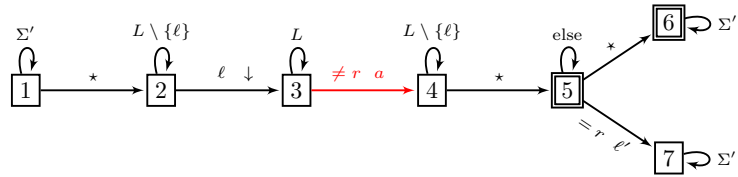
\begin{center}
\scalebox{.9}{
\begin{tikzpicture}[>=latex',shorten >=1pt,node distance=1.5cm,on grid,auto, thick,
roundnode/.style={circle, draw=black, thick, minimum size=2mm},
squarednode/.style={rectangle, draw=black, thick, minimum size=3mm}]
 
\node[squarednode] (1)   {$1$}; 
\node[squarednode] (2)  [right =2cm of 1] {$2$};
\node[squarednode] (3)  [right=2cm of 2] {$3$};
\node[squarednode] (4)  [right =2cm of 3] {$4$};
\node[squarednode,accepting] (5)  [right =2cm of 4] {$5$};
\node[squarednode,accepting] (6)  [above right=1cm and 1.5cm  of 5] {$6$};
\node[squarednode] (7)  [below =2cm of 6] {$7$};

\path [->] (1) edge [loop above] node [above] {\scriptsize{$\alphabet'$ }} ();
\path [->] (1) edge node [midway,above] {\scriptsize{$\star$ }} (2);
\path [->] (2) edge [loop above] node [above] {\scriptsize{$\locs\setminus\{\loc\}$ }} ();
\path [->] (2) edge node [sloped,above] {\scriptsize{$~~\loc~~\downarrow$}} (3);
\path [->] (3) edge [loop above] node [above] {\scriptsize{$\locs$ }} ();
\path [->,red] (3) edge node [sloped,above,red] {\scriptsize{$\neq r~~a~~$}} (4);
\path [->] (4) edge [loop above] node [above] {\scriptsize{$\locs\setminus\{\loc\}$ }} ();
\path [->] (4) edge node [sloped,above] {\scriptsize{$\star$ }} (5);
\path [->] (5) edge [loop above] node [above]  {\scriptsize{else}} ();
\path [->] (5) edge node [sloped,above] {\scriptsize{$\star$ }} (6);
\path [->] (5) edge node [sloped,below] {\scriptsize{$=r~~\loc'~~$}} (7);
\path [->] (6) edge [loop right] node [right] {\scriptsize{$\alphabet'$ }} (6);
\path [->] (7) edge [loop right] node [right] {\scriptsize{$\alphabet'$ }} (7);
\end{tikzpicture}
}
\end{center}

\item Family~$\Rfamily_7$: for all transitions $\loc\ttto{\neq r}{a}{r}\loc'$, we add a~$1$-NRA that
	only accepts data words such that  one block contains some $(\loc,x)$
	and $(a,d)$ with $x\neq d$ where the next block does not have~$(\loc',d)$.
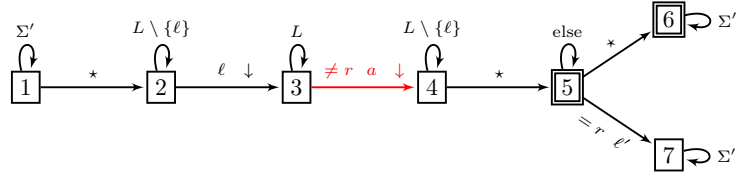
\begin{center}
\scalebox{.9}{
\begin{tikzpicture}[>=latex',shorten >=1pt,node distance=1.5cm,on grid,auto, thick,
roundnode/.style={circle, draw=black, thick, minimum size=2mm},
squarednode/.style={rectangle, draw=black, thick, minimum size=3mm}]
 
\node[squarednode] (1)   {$1$}; 
\node[squarednode] (2)  [right =2cm of 1] {$2$};
\node[squarednode] (3)  [right=2cm of 2] {$3$};
\node[squarednode] (4)  [right =2cm of 3] {$4$};
\node[squarednode,accepting] (5)  [right =2cm of 4] {$5$};
\node[squarednode,accepting] (6)  [above right=1cm and 1.5cm  of 5] {$6$};
\node[squarednode] (7)  [below =2cm of 6] {$7$};

\path [->] (1) edge [loop above] node [above] {\scriptsize{$\alphabet'$ }} ();
\path [->] (1) edge node [midway,above] {\scriptsize{$\star$ }} (2);
\path [->] (2) edge [loop above] node [above] {\scriptsize{$\locs\setminus\{\loc\}$ }} ();
\path [->] (2) edge node [sloped,above] {\scriptsize{$~~\loc~~\downarrow$}} (3);
\path [->] (3) edge [loop above] node [above] {\scriptsize{$\locs$ }} ();
\path [->,red] (3) edge node [sloped,above,red] {\scriptsize{$\neq r~~a~~\downarrow$}} (4);
\path [->] (4) edge [loop above] node [above] {\scriptsize{$\locs\setminus\{\loc\}$ }} ();
\path [->] (4) edge node [sloped,above] {\scriptsize{$\star$ }} (5);
\path [->] (5) edge [loop above] node [above]  {\scriptsize{else}} ();
\path [->] (5) edge node [sloped,above] {\scriptsize{$\star$ }} (6);
\path [->] (5) edge node [sloped,below] {\scriptsize{$=r~~\loc'~~$}} (7);
\path [->] (6) edge [loop right] node [right] {\scriptsize{$\alphabet'$ }} (6);
\path [->] (7) edge [loop right] node [right] {\scriptsize{$\alphabet'$ }} (7);
\end{tikzpicture}
}
\end{center}
\end{enumerate}	

The proof is complete. 
\end{proof}

By Lemmas~\ref{lemma_nonuniv_to_synch} and~\ref{lemma_reduction_synch_univ} and  $\ackermann$-completeness of the non-universality problem for $1$-NRA,  which follows from Theorem 2.7 and the proof of Theorem 5.2 in~\cite{DBLP:journals/tocl/DemriL09}, and the result for counter automata with incrementing errors in~\cite{DBLP:conf/lics/FigueiraFSS11},
we obtain the following theorem.

\begin{theorem}
	The synchronization problem for $1$-NRAs is $\ackermann$-complete.	
\end{theorem}

\section{Length-Bounded synchronizing data words for NRA\lowercase{s}}\label{short}

As proved in the previous section, the synchronization problem for NRAs is in general undecidable. 
In this section, we study the length-bounded synchronization problem for NRAs, in which 
 the synchronizing data words are required to be shorter than a given length (written in binary).

To decide the synchronization problem in $1$-RAs, both in the deterministic and nondeterministic setting, 
we rely on Lemma~\ref{lemmafiniteD}. With this lemma at hand, it was sufficient to search for synchronizing data words that 
first input a single datum~$x$ (chosen arbitrary) as many times as necessary
to have the set of successor configurations included in~$\locs\times \{x\}$. In the next step, this obtained set of successor configurations 
was synchronized in a singleton. However, the shortest synchronizing data words do not always follow this pattern, for 
an example see Figure~\ref{fig:boundedRA}.
Observe that the  data word~$(a,x)(b,y)(b,z)$ is synchronizing with length~$3$ (not exceeding the bound~$3$).
However, all synchronizing data words that repeat a datum such as~$x$, to first bring the RA to a finite set, have length at least~$4$.
The example shows that one cannot rely on the techniques developed in Section~\ref{synchNRAs} to 
decide the length-bounded synchronization problem for NRA.

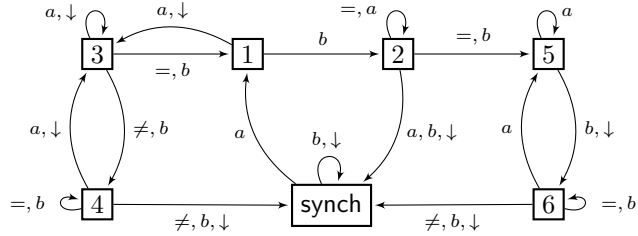
\begin{figure}[t]
\begin{center}

\begin{tikzpicture}[>=latex',shorten >=1pt,node distance=2cm,on grid,auto, 
		roundnode/.style={circle, draw=black,  thick, minimum size=5mm},
squablacknode/.style={rectangle, draw=black,  thick, minimum size=4mm}]

\node[squablacknode] (1) at (0,0) {};
\node at (0,0) {$1$};
\node[squablacknode] (6) [right=of 1] {};
\node [right=of 1] {$2$};

\node[squablacknode] (2) [right =of 6] {};
\node [right =of 6] {$5$};
\node[squablacknode] (4) [below =of 2] {};
\node[below =of 2] {$6$};
\node[squablacknode] (3) [left =of 1] {};
\node [left =of 1] {$3$};
\node[squablacknode] (5) [below =of 3] {};
\node [below =of 3] {$4$};

\node[squablacknode] (7) [below right=of 1, xshift=-0.3cm, yshift=-.6cm] {$\mathsf{synch}$};

\path [->] (3) edge  node [sloped, midway, below] {\scriptsize{$=, b$}}  (1) ;
\path [->] (1) edge [bend right] node [sloped, midway, above] {\scriptsize{$a, \downarrow$}}  (3) ;

\path [->] (6) edge node [sloped, midway, above] {\scriptsize{$=, b$}}  (2) ;
\path [->] (1) edge  node [sloped, midway, above] {\scriptsize{$b$}}  (6) ;
\path [->] (6) edge [bend left] node [right] {\scriptsize{$a,b, \downarrow$}}  (7);
\path [->] (7) edge [bend left]  node [left] {\scriptsize{$a$}}  (1);

\path [->] (4) edge  node [sloped, midway, below] {\scriptsize{$\neq, b, \downarrow$}}  (7);
\path [->] (5) edge  node [sloped, midway, below] {\scriptsize{$\neq, b, \downarrow$}}  (7);

\path [->] (2) edge [bend left] node [right] {\scriptsize{$b, \downarrow$}}  (4);
\path [->] (4) edge [bend left] node [left] {\scriptsize{$a$}}  (2);

\path [->] (3) edge [bend left] node [right] {\scriptsize{$\neq, b$}}  (5);
\path [->] (5) edge [bend left] node [left] {\scriptsize{$a, \downarrow$}}  (3);
\path[->] (3) edge [out=115, in=75,looseness=8]  node [left=0.1cm] {\scriptsize{$a, \downarrow$}} (3);

\path[->] (6) edge [out=115, in=75,looseness=8]  node [left=0.1cm] {\scriptsize{$=, a$}} (6);

\path[->] (2) edge [out=115, in=75,looseness=8]  node [right=0.1cm] {\scriptsize{$a$}} (2);

\path[->] (5) edge [out=195, in=165,looseness=8]  node [left=0.1cm] {\scriptsize{$=, b$}} (5);
\path[->] (4) edge [out=-15, in=15,looseness=8]  node [right=0.1cm] {\scriptsize{$=, b$}} (4);
\path[->] (7) edge [out=115, in=75,looseness=6]  node [sloped, midway, above] {\scriptsize{$b,  \downarrow$}} (7);
\end{tikzpicture}

\end{center}
\caption{An RA with synchronizing data word $(a,x)(b,y)(b,z)$ with three distinct data values $x,y,z$. 
	The approach of using a unique data value to shrink the infinite set of 
	configurations to a finite subset only yields synchronizing data words of length greater than 3. 
   } 
\label{fig:boundedRA}
\end{figure}

\medskip

In this section, we prove 
\begin{theorem} \label{theo-nexp}
	The  length-bounded synchronization problem for NRAs is $\nexptime$-complete.
\end{theorem}
The $\nexptime$-membership of the length-bounded synchronization problem is straightforward: guess a data word~$w$ 
shorter than the given length (that is written in binary and thus may be  exponential in the length)  
and check in $\exptime$ whether $w$ is synchronizing. 
Our main contribution  is to prove the $\nexptime$-hardness of this problem, for which in turn,
by Lemma~\ref{lemma_nonuniv_to_synch},  it is sufficient to show that the \emph{length-bounded universality}  problem
is  co-$\nexptime$-complete.
The length-bounded universality  problem  asks, given an RA and 
$N\in \nat$ encoded in binary, whether all data words~$w$ with $\abs{w}\leq N$ are  in the language of the automaton.

\begin{theorem}
	The length-bounded  universality problem for NRAs is  co-$\nexptime$-complete.
\end{theorem}
\begin{proof}
The length-bounded universality problem for NRAs can be solved in co-$\nexptime$, by guessing a
(possibly exponentially long) data word, and check whether the guessed word is a witness
for non-universality of the RA.

We prove that the complement of the length-bounded universality problem is $\nexptime$-hard. 
The proof is a reduction from the  \emph{membership problem of  $\mathcal{O}(2^n)$-time bounded nondeterministic Turing machines}: 
given a nondeterministic Turing machine $\tm$ and an input word $x$, decide whether $\tm$ accepts $x$ within time bound~$2^{\abs{x}}$. 
This problem is  $\nexptime$-complete.

Given a nondeterministic Turing machine $\tm$ and an input $x$ of length $n$, 
we construct an NRA $\R$ equipped with an initial location and a set of accepting locations, 
and a bound $N$ (encoded in binary) such that 
there exists a witness of non-universality $w$ (i.e., $w\not\in L(\R)$) with $\abs{w}\leq N$   if, and only if, 
$\tm$ has some accepting computation on $x$ within time bound~$2^n$.

Let $\tm$ have the set $Q$ of control states and the tape alphabet $\Gamma$. 
Let us recall that a configuration  of~$\tm$ is a word in the language $\Gamma^* (Q\times\Gamma) \Gamma^*$, where each letter in $(Q\times\Gamma) \cup \Gamma$ encodes a single cell and the position of the reading/writing head.
A computation~$\rho$ of $\tm$ is a sequence $c_0c_1c_2 \cdots$ of configurations that respects the transition function of 
the Turing machine. 
Without loss of generality, we assume that $\tm$ has a self-loop on  all accepting states.
Hence for the input~$x\in \Gamma^{*}$ of length~$n$, 
all accepting computations~$\rho$ of $\tm$ are sequences of length exactly $2^n$, and 
all configurations $c_i$ along such a computation are words $c_i\in \Gamma^* (Q\times\Gamma) \Gamma^* $ of length at most $2^{n}$. 
In the following, we pad the configurations shorter than $2^n$  with $\Hsquare$ at the tail such that the length of all 
such configurations become equal to $2^n$.

Let $\Sigma_\tm := \Sigma \cup \Sigma'$, where 
\[\Sigma=(Q\times\Gamma) \cup \Gamma \cup \dot{(Q\times\Gamma)} \cup \dot{\Gamma} \cup \{\Hsquare,\dot{\Hsquare},\#,\star\}\]
be such that  $\Hsquare,\dot\Hsquare,\#,\star \not\in \Gamma$.
Here, $\dot{(Q\times \Gamma)}$ and $\dot{\Gamma}$ denote a \emph{dotted} version of letters in $Q\times \Gamma$ and $\Gamma$; formally 
\[\{\dot{(q,a)}\mid (q,a)\in (Q\times\Gamma)\} \text{ \quad \quad   and  \quad \quad  } \{\dot a\mid a\in\Gamma\},\]
and $\Sigma'$ will be defined later.
Let $K=2^{3n}+2^{2n}+1$.
Given a computation~$\rho= c_1 \cdots c_{2^n}$,
we define $u(\rho)\in\Sigma^{K}$, roughly speaking, such that

\begin{enumerate}
\item It consists of $2^{n}$ copies of $\rho$ (with some extra delimiters).
	\item Between all consecutive copies of $\rho$ there is a $\star$ delimiter, and $u(\rho)$ starts and ends with  $\star$, too.
Hence,  there are $2^n+1$ occurrences of $\star$ in $u(\rho)$.
	\item In each copy of $\rho$, there is a $\#$ delimiter between consecutive configurations. 
	Since there are $2^n$ configurations in (each copy of)~$\rho$, the number of $\#$ in $u(\rho)$ is $2^n(2^n-1)$.

	\item In the $i$-th copy of  $\rho$, the letter for the $i$-th cell of every participating configuration $c_i$ is dotted, 
	all other letters are non-dotted. Hence, in each copy of $\rho$ 
	there are exactly $2^n$ dotted letters (one in each configuration of  $\rho$), with distance $2^n+1$.
	\item The distance between two $\star$ delimiters is $2^{2n}+2^n-1$, due to the fact that $\rho$ consists of $2^n$
	configurations,  each of which has $2^n$ tape cells in turn and is separated from the next configuration by a $\#$ delimiter.	
\end{enumerate}

\begin{figure}[t]
\centering
  \scalebox{.9}{
\begin{tikzpicture}[->,>=stealth',shorten >=1pt,auto,node distance=4cm,thick,node/.style={circle,draw,scale=0.9}, roundnode/.style={circle, draw=black, thick, minimum size=6mm},]
\tikzset{every state/.style={minimum size=0pt}};

 \fill [gray!30] (-4.1,2.2) rectangle (4.3,1.3);
 
 \fill [gray!30] (-4.1,1.05) rectangle (4.3,0.15);
 
 \fill [gray!30] (-4.1,-0.1) rectangle (4.3,-1);
 
 \fill [gray!30] (-4.1,-1.3) rectangle (4.3,-2.15);

\node (table) at (0,0) {\setlength{\tabcolsep}{2mm}
  \begin{tabular}{ccccccccccc}
  $\star$ & $\textcolor{magenta}{\dot{(q_\init,a_1)}}$ & $a_2$ & $\Hsquare$ & $\Hsquare$ & $\#$ & $\textcolor{magenta}{\dot a}$ & $(q,a_2)$ & $\Hsquare$ & $\Hsquare$ & $\dots$ \\
  $0$ & $\textcolor{blue}{1}$ & $\textcolor{blue}{2}$ & $\textcolor{blue}{3}$ & $\textcolor{blue}{4}$ & $0$ & $\textcolor{blue}{6}$ & $\textcolor{blue}{7}$ & $\textcolor{blue}{8}$ & $\textcolor{blue}{9}$ &\\ 
  \\
  $\star$ & $(q_\init,a_1)$ & $\textcolor{magenta}{\dot{a_2}}$ & $\Hsquare$ & $\Hsquare$ & $\#$ & $a$ & $\textcolor{magenta}{\dot{(q,a_2)}}$ & $\Hsquare$ & $\Hsquare$ & $\dots$\\
  $0$ & $1$ & $2$ & $3$ & $4$ & $0$ & $6$ & $7$ & $8$ & $9$ &\\
  \\
  $\star$ & $(q_\init,a_1)$ & $a_2$ & $\textcolor{magenta}{\dot \Hsquare}$ & $\Hsquare$ & $\#$ & $a$ & $(q,a_2)$ & $\textcolor{magenta}{\dot \Hsquare}$ & $\Hsquare$ & $\dots$\\
  $0$ & $1$ & $2$ & $3$ & $4$ & $0$ & $6$ & $7$ & $8$ & $9$ &\\
   \\
  $\star$ & $(q_\init,a_1)$ & $a_2$ & $ \Hsquare$ & $\textcolor{magenta}{\dot\Hsquare}$ & $\#$ & $a$ & $(q,a_2)$& $\Hsquare$ & $\textcolor{magenta}{\dot \Hsquare}$ & $\dots$\\
  $0$ & $1$ & $2$ & $3$ & $4$ & $0$ & $6$ & $7$ & $8$ & $9$ &\\
\end{tabular}};

\node  at (-5.5,2)  {\begin{tabular}{ll} \textcolor{gray}{copies of the} \\ \textcolor{gray}{encodings}\end{tabular}};

\node (1) at (-5,2) {};
\node (2) at (-4,1.6) {};
\node (3) at (-4.1,0.5) {};

\path [->] (1) edge[gray] (2);
\path [->] (1) edge[gray] (3);

\node  at (-.5,3)  {\begin{tabular}{ll} \textcolor{magenta}{shifting} \\ \textcolor{magenta}{dotted letters}\end{tabular}};

\node (1') at (-1.3,3.1) {};
\node (2') at (-2.8,2.1) {};
\path [->] (1') edge[magenta] (2');

\node (3') at (-1,2.8) {};
\node (4') at (-1.5,0.7) {};
\path [->] (3') edge[magenta] (4');

\node  at (5.4,1)  {\begin{tabular}{ll} \textcolor{blue}{unique data} \\ \textcolor{blue}{identifiers}\end{tabular}};

\node (1'') at (4.7,1) {};
\node (2'') at (3.2,1.6) {};

\path [->] (1'') edge[blue] (2'');

	\end{tikzpicture}
} 
\caption{
	Partial encoding of~$u(\rho)$ for an accepting $2^2$-time bounded computation~$\rho$ of a Turing machine on $a_1 a_2$.
}
\label{fig:bound_encoding}
\end{figure}
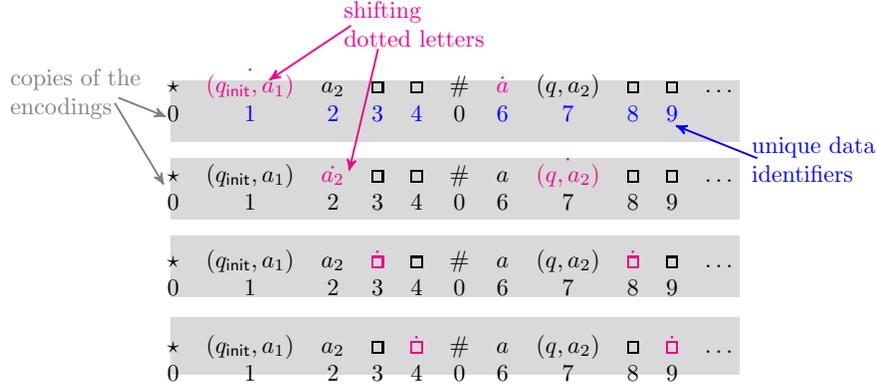

Figure~\ref{fig:bound_encoding} illustrates an example of $u(\rho)$. Observe that for all~$\rho=c_0 c_1 \cdots c_{2^n}$, 
we have
\[\abs{u(\rho)}= \underbrace{2^n}_{\text{number of copies of }\rho} \overbrace{2^n}^{\text{number of configurations }c_i \text{ in } \rho} 
\underbrace{2^n}_{\text{length of } c_i } + \overbrace{2^{n}(2^n-1)}^{\text{number of  } \# }+\underbrace{(2^n+1)}_{\text{number of } \star} =K. \]

We define a data language $\lang$  over the alphabet $\Sigma$ such that data words in this language are 
faithful encodings of computations $\rho$ of $\tm$
over the input word~$x$. In particular, the language contains all data words~$v$ that satisfy the following conditions:

\begin{enumerate}
		\addtocounter{enumi}{5} 
	\item  Let $\proj(v)$ be the projection of $v$ into $\Sigma$ (i.e., omitting the data values).
	There exists some accepting computation~$\rho$ of $\tm$ on the input $x$ such that 
	$\proj(v)=u(\rho)$.
	\item  The letters $\star$ and $\#$ occur only with a unique datum, say datum $0$ (and no other letter occurs with that datum).
	\item For all  occurrences of $\star$, for all $1\leq i\leq 2^{2n}+2^n-1$, all letters at the $i$-th positions
	after each $\star$ must carry the same datum, say datum $i$. Except for occurrences of $\#$, the datum~$i$ is exclusive for the $i$-th positions
	after occurrences of $\star$.

\end{enumerate}

Given a data word~$v\in \lang$ such that $\proj(v)=u(\rho)$ for some computation~$\rho$, condition $(8)$ and previous conditions on $u(\rho)$  entail that
for all $1\leq j,k\leq 2^n$
the $j$-th tape cell in the $k$-th configuration $c_k$ of all copies of $\rho$ in $v$ carries the same datum
(revisit Figure~\ref{fig:bound_encoding}).
Observe that  all data words $v\in \lang$ use exactly $2^{2n} + 1$ distinct data values.

By definition of $\lang$, we see that $\lang$ is non-empty if, and only if, there is an accepting computation~$\rho$ of
$\tm$ over~$x$. 
Recall that $\Sigma_{\tm}=\Sigma \cup \Sigma'$ (where $\Sigma'$ is defined later).
Below, we construct a $1$-NRA $\R$ over alphabet $\Sigma_{\tm}$ such 
that the language accepted by $\R$ (projected into $\Sigma$, ignoring $\Sigma'$ letters) is    
 the complement of $\lang$. 
At the end, we  examine the existence of $N\in \mathcal{O}(K)$ such that
$\tm$ has an accepting computation over $x$ if, and only if, $\R$ is (length-bounded) non-universal 
with respect to the bound~$N$. 

The $1$-NRA $\R$ is the union of several $1$-NRAs and DFAs that we describe in the following. 
Each of these automata violates one of the necessary conditions for data words $v$ to be in $\lang$. 

\begin{itemize}
\item We add a DFA that accepts data words $v$ such that $\proj(v)$ is not in the regular language  
$(\star L)^* \star$ where $L$ is
defined by
$$ \left(  (\Gamma + \dot\Gamma)^*  \, \left( (Q\times\Gamma) +
\dot{(Q\times\Gamma)} \right) \, (\Gamma + \dot\Gamma)^*\, (\Hsquare + \dot \Hsquare)^*\, \# \right)^*.$$
\item We add a DFA that accepts data words $v$ such that $\proj(v)$ does not start with 
$$\star \left(\dot{(q_\init,a_1 )}a_2 a_3\dots a_n  \Hsquare^*  \#\right), $$
where $q_\init$ is the initial control state of $\tm$ and $x=a_1 a_2 \cdots a_n$ is the input. 
This regular expression also guarantees  that in the first copy of $\rho$, the first cell is dotted.
\item We add a DFA that accepts data words $w$ containing at least two dotted letters between two consecutive $\#$. 
\item  We add a~$1$-NRA that accepts data words in which some delimiter occurs with some datum different from the datum for the first $\star$.
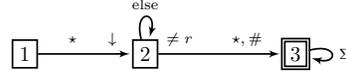
\begin{center}
\scalebox{.8}{
\begin{tikzpicture}[>=latex',shorten >=1pt,node distance=1.5cm,on grid,auto, thick,
roundnode/.style={circle, draw=black, thick, minimum size=2mm},
squarednode/.style={rectangle, draw=black, thick, minimum size=3mm}]
 
\node[squarednode] (1)   {$1$}; 
\node[squarednode] (2)  [right =2cm of 1] {$2$};
\node[squarednode,accepting] (3)  [right=2.5cm of 2] {$3$};

\path [->] (1) edge node [midway,above] {\scriptsize{$~~\star~~~~\downarrow$ }} (2);
\path [->] (2) edge [loop above] node [above] {\scriptsize{else}} ();
\path [->] (2) edge node [sloped,above] {\scriptsize{$\neq r$ \quad \quad $\star,\#~~$}} (3);
\path [->] (3) edge [loop right] node [right] {\scriptsize{$\alphabet$ }} (3);
\end{tikzpicture}
}
\end{center}

\item We add a $1$-NRA that accepts data words in which some other letter  appears with 
the datum  dedicated to delimiters $\star$  and $\#$. 
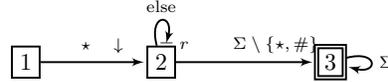
\begin{center}
\scalebox{.9}{
\begin{tikzpicture}[>=latex',shorten >=1pt,node distance=1.5cm,on grid,auto, thick,
roundnode/.style={circle, draw=black, thick, minimum size=2mm},
squarednode/.style={rectangle, draw=black, thick, minimum size=3mm}]
 
\node[squarednode] (1)   {$1$}; 
\node[squarednode] (2)  [right =2cm of 1] {$2$};
\node[squarednode,accepting] (3)  [right=2.5cm of 2] {$3$};

\path [->] (1) edge node [midway,above] {\scriptsize{$~~\star~~\downarrow$ }} (2);
\path [->] (2) edge [loop above] node [above] {\scriptsize{else}} ();
\path [->] (2) edge node [sloped,above] {\scriptsize{$=r$ \quad \quad $\alphabet\setminus\{\star,\#\}~~$}} (3);
\path [->] (3) edge [loop right] node [right] {\scriptsize{$\alphabet$ }} (3);
\end{tikzpicture}
}
\end{center}

\item We add $1$-NRA that accepts data words in which there are two letters (other than $\#$) between two 
consecutive~$\star$ that carry the same datum. 
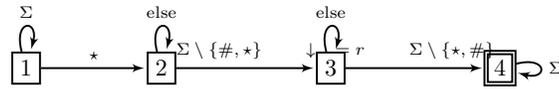
\begin{center}
\scalebox{.9}{
\begin{tikzpicture}[>=latex',shorten >=1pt,node distance=1.5cm,on grid,auto, thick,
roundnode/.style={circle, draw=black, thick, minimum size=2mm},
squarednode/.style={rectangle, draw=black, thick, minimum size=3mm}]
 
\node[squarednode] (1)   {$1$}; 
\node[squarednode] (2)  [right =2cm of 1] {$2$};
\node[squarednode] (3)  [right =2.5cm of 2] {$3$};
\node[squarednode,accepting] (4)  [right=2.5cm of 3] {$4$};

\path [->] (1) edge [loop above] node [above] {\scriptsize{$\alphabet$}} (1);
\path [->] (1) edge node [midway,above] {\scriptsize{$\star$}} (2);
\path [->] (2) edge node [midway,above] {\scriptsize{$ \alphabet\setminus \{\#,\star\}$ \quad \quad $\downarrow$}} (3);
\path [->] (2) edge [loop above] node [above] {\scriptsize{else}} (2);
\path [->] (3) edge node [sloped,above] {\scriptsize{$=r$ \quad \quad $\alphabet\setminus\{\star,\#\}$}} (4);
\path [->] (3) edge [loop above] node [above] {\scriptsize{else}} (3);
\path [->] (4) edge [loop right] node [right] {\scriptsize{$\alphabet$ }} (4);
\end{tikzpicture}
}
\end{center}

\item We add a $1$-NRA that accepts data words $v$ such that
there are two  consecutive $\#$ whose distance  is not exactly~$2^{n}$ (ignoring the occurrences of $\star$). 
For this we use a variant of $\R_{\counter(n)}$ implementing a binary counter introduced in Section \ref{synchNRAs}. 
For accepting data words $v$ such that the distance between two consecutive $\#$ is \emph{less than} $2^{n}$, we add a transition 
$$2^{n}_c \xrightarrow{\, \# \,} \loc_\acc,$$
and for accepting those words  that the distance is \emph{more than} $2^n$, we add a transition 
$$2^{n} \xrightarrow{\, \Sigma \,} \loc_\acc.$$
Here, $\loc_\acc$ is an accepting location with a self-loop for every letter in $\Sigma$.
\end{itemize}
For the next four $1$-NRAs we can use simple variants of $\R_{\counter(n)}$: 
\begin{itemize}
\item We add a $1$-NRA that accepts data words $v$
 such that  between two consecutive $\star$, the letter $\#$ does not occur exactly $2^n-1$ times. 
\item We add a $1$-NRA  that accepts data words $v$ such that  $\star$ does not occur exactly $2^n + 1$ times. 
\item We add a $1$-NRA that  accepts data words $v$ such that
the distance between two consecutive  dotted letters is 
  not exactly $2^n+1$, if no delimiter $\star$ is seen between these two letters.
We add another $1$-NRA that  accepts data words $v$ such that
the distance between two consecutive  dotted letters is 
 not exactly $2^n+2$ if  $\star$ is seen.
\item We add a $1$-NRA that  accepts data words $v$ such that
the letters with $2^{2n}+2^n-1$ distance carry different data.

\end{itemize}

To implement the above binary counters with $1$-NRAs, we finally define 
\[\Sigma'=\{\Bit{i}^{d},\Bit{i}^\#,\Bit{i}^\star, \dot{\Bit{i}}, \Bit{i}^{x}, \Bit{i+n}^{x}\mid 0\leq i \leq n\},\]
where
\begin{itemize}
	\item letters $\Bit{0}^{d},\dots,\Bit{n}^{d}$ for counting the distance between two consecutive $\#$.
	The counter takes into account only letters in $\Sigma \setminus\{\star\}$, ignoring the occurrences of $\star$ and other $\Bit{i}$-letters from $\Sigma'$.
	The $1$-NRA detects whether the distance is less or greater than  $2^n$.
	\item letters $\Bit{0}^\#, \dots, \Bit{n}^\#$ for counting the occurrences of $\#$. The $1$-NRA 
	detects whether the number of $\#$ between two consecutive $\star$ is less or greater than $2^n-1$.
	\item letters $\Bit{0}^\star,\dots,\Bit{n}^\star$ for counting the occurrences of $\star$ (to check against  $2^n+1$).
	\item letters $\dot{\Bit{0}},\dots,\dot{\Bit{n}}$ for counting the distance between two consecutive dotted letters
	(to check against $2^n+1$ or $2^n+2)$.
\item letters $\Bit{0}^{x},\dots,\Bit{2n}^{x}$ for counting the distance between two letters that carry the same datum
	(to check against~$2^{2n}+2^n+1$).

\end{itemize}
We construct all these gadgets such that the $\Bit{}$-letters always
carry the same datum as the delimiters.

The union of all above $1$-NRAs and DFAs accepts all data words except those  
$v$ such that $\proj(v)=(\star \rho)^{2^n} \star$ 
(that in addition respect the uniqueness conditions on data appearing in~$v$).
Finally, we add NRAs that check whether $\rho= c_1 \cdots c_{2^n}$ in such $v$ is not a 
faithful computation of $\tm$, or it is not an \emph{accepting} computation.
To this aim, for all words $\sigma_1\sigma_2\sigma_3\in ((Q\times\Gamma)\cup\Gamma)^3$ of length three 
such that $\sigma_1\sigma_2\sigma_3$ can appear at some position $i$ in a  valid configuration~$c$ of  $\tm$, 
we define $Post(\sigma_1\sigma_2\sigma_3)$ to be the set of words 
$u\in((Q\times\Gamma)\cup\Gamma)^3$ that can  appear in a successor configuration of~$c$ in the same position~$i$ (according to the rules of $\tm$). 

\begin{itemize}
\item 
For all words $\dot{\sigma}_1\sigma_2\sigma_3\in (\dot{Q\times\Gamma})\cup\dot{\Gamma})((Q\times\Gamma)\cup\Gamma)^2$
that starts with a dotted letter,
we add a $1$-NRA that accepts data words that  for some occurrence of the subword  $(\dot{\sigma_1},d_1)(\sigma_2,d_2)(\sigma_3,d_3)$ 
with some data $d_1,d_2,d_3$,
 the subword $\tau_1\dot{\tau_2}\tau_3$ (ignoring the data values)
 with exactly $2^{2n}+2^{n+1}+1$ distance is not in $Post(\sigma_1\sigma_2\sigma_3)$.
Observe that the subword $\dot{\sigma_1}\sigma_2\sigma_3$ is intuitively
indicating some part of  some configuration~$c$ in some copy of $\rho$, and 
$\tau_1\dot{\tau_2}\tau_3$ with distance $2^{2n}+2^{n+1}+1$ is a subword of the
	successor configuration of $c$ in the next copy of $\rho$.  		
	
 The following NRA is for the case $(q_{\init},a_1)a_2\Hsquare$. To implement this 
$1$-NRA, we rely on the previous conditions that two letters (apart from the delimiters)
with the same datum have the exact distance $2^{2n}+2^{n+1}+1$ (checked with a parallel $1$-NRA).
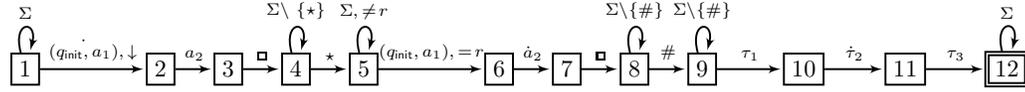
\begin{center}
\scalebox{.9}{
\begin{tikzpicture}[>=latex',shorten >=1pt,node distance=1.5cm,on grid,auto, thick,
roundnode/.style={circle, draw=black, thick, minimum size=2mm},
squarednode/.style={rectangle, draw=black, thick, minimum size=3mm}]
 
\node[squarednode] (1)   {$1$}; 
\path [->] (1) edge [loop above] node [above] {\scriptsize{$\Sigma$ }} (1);

\node[squarednode] (2) at (2,0)  {$2$}; 
\path [->] (1) edge node [above] {\scriptsize{$\dot{(q_\init,a_1)},\downarrow$ }} (2);

\node[squarednode] (3) at (3,0)  {$3$}; 
\path [->] (2) edge node [above] {\scriptsize{$a_2$ }} (3);

\node[squarednode] (4) at (4,0)  {$4$}; 
\path [->] (3) edge node [above] {\scriptsize{$\Hsquare$ }} (4);

\path [->] (4) edge [loop above] node [above] {\scriptsize{$\Sigma\backslash\ \{\star\}$ }} (4);

\node[squarednode] (5) at (5,0)  {$5$}; 
\path [->] (4) edge node [above] {\scriptsize{$\star$ }} (5);

\path [->] (5) edge [loop above] node [above] {\scriptsize{$\Sigma,\neq\!r$ }} (5);

\node[squarednode] (6) at (7,0)  {$6$}; 
\path [->] (5) edge node [above] {\scriptsize{$(q_\init,a_1),=\!r$ }} (6);

\node[squarednode] (7) at (8,0)  {$7$}; 
\path [->] (6) edge node [above] {\scriptsize{$\dot a_2$ }} (7);

\node[squarednode] (8) at (9,0)  {$8$}; 
\path [->] (7) edge node [above] {\scriptsize{$\Hsquare$ }} (8);

\path [->] (8) edge [loop above] node [above] {\scriptsize{$\Sigma\backslash\{\#\}$ }} (8);

\node[squarednode] (9) at (10,0)  {$9$}; 
\path [->] (8) edge node [above] {\scriptsize{$\#$ }} (9);
\path [->] (9) edge [loop above] node [above] {\scriptsize{$\Sigma\backslash\{\#\}$ }} (9);

\node[squarednode] (10) at (11.5,0)  {$10$}; 
\path [->] (9) edge node [above] {\scriptsize{$\tau_1$ }} (10);

\node[squarednode] (11) at (13,0)  {$11$}; 
\path [->] (10) edge node [above] {\scriptsize{$\dot\tau_2$ }} (11);

\node[squarednode,accepting] (12) at (14.5,0)  {$12$}; 
\path [->] (11) edge node [above] {\scriptsize{$\tau_3$ }} (12);
\path [->] (12) edge [loop above] node [above] {\scriptsize{$\Sigma$ }} (12);
\end{tikzpicture}
}
\end{center}
\item We add a DFA that accepts data words $v$ such that the last 
configuration in $\rho$ does not contain  a letter in 
$(Q_{\acc}\times\Gamma) \cup \dot{(Q_{\acc}\times\Gamma)}$, where $Q_\acc$ is the set 
of accepting control states of $\tm$. 
\end{itemize}

To complete the proof, we 
examine the existence of $N\in \mathcal{O}(K)$ such that
$\tm$ has an accepting computation over $x$ if, and only if, $\R$ is (length-bounded) non-universal 
with respect to the bound~$N$.
 Given the shortest witness~$w\in \Sigma_{\tm}^+$ of non-universality of $\R$, the projection~$v$ of
$w$ into $\Sigma$ encodes an accepting computation of $\tm$ over $x$, and subsequently has
length exactly $K$. The extra letters of $w$ compared to~$v$ are to implement the  five needed counters faithfully. 
However, these letters do not  increase the length of $w$ much more than~$K$: 
for instance, 
the condition for counting the occurrences of $\#$ requires that we accompany 
every $\#$ with a single $\Bit{i}^\#$-letter.
Hence, 

\[
N \leq \left(
\overbrace{2^{3n+1}}^{\text{cell letters and $\Bit{i}^d$ } }
+
\underbrace{2^{n+1}(2^n-1)}_{\text{$\#$ and $\Bit{i}^\#$ } } 
+ 
\overbrace{2^{n+1}+2}^{\text{$\star$ and $\Bit{i}^\star$ } }
+
\underbrace{2^{3n+1}}_{\text{cell letters and $\dot{\Bit{i}}$}}
+
\overbrace{2^{3n+1}}^{\text{cell letters and $\Bit{i}^x$}}
\right).
\]
Note that $N$ is still exponential in $n$.

The construction of $\R$ is complete and the $\nexptime$-hardness follows from the sketched reduction. 
Note that the result already holds for $1$-NRAs.
\end{proof}

There is a natural reduction from the non-universality problem for $1$-NRAs to the 
emptiness problem for single-register alternating RAs ($1$-ARAs). 
The trivial $\nexptime$ membership (guess and check) and Theorem~\ref{theo-nexp} lead to  
 the $\nexptime$-completeness of the 
length-bounded  emptiness problem for $1$-ARAs.

\medskip
\subparagraph*{{\bf Acknowledgements}}
We thank Sylvain Schmitz 
for helpful discussions on well-structured systems and non-elementary complexity classes.
We thank James Worrell for  inspiring discussions, especially drawing our attention to a 
trick that  simplified  the $\nexptime$-hardness construction.
We appreciate  the anonymous reviewers for  their  insightful comments and suggestions.


\newpage
\section*{Appendix}
\section{Proofs for Deterministic Register Automata}
\label{append_DRAs}
\begin{qlemma}{\ref{lemmasyncD}}
\ \lemmasyncD
\end{qlemma}

\begin{proof}
Let~$\R =\tuple{\locs,\reg,\alphabet,T}$ be a DRA on the data domain~$\domain$ and with $k\geq 1$ registers.
Recall that we denote by $\data(w)$ the data occurring 
in data words~$w$; for configurations $q=(\loc,\valuation)$ we use the same notation 
$\data(q)=\{\valuation(r)\mid r \in \reg\}$ to denote the 
data appearing in the valuation of~$q$. 
Let $\pi:Y_1 \to Y_2$ be a bijection on data where $Y_1,Y_2\subseteq \domain$.
For every configuration $q=(\loc,\nu)$, define $\pi(q)=(\loc,\nu')$, where $\nu'$ satisfies $\nu'(r)=\pi(\nu(r))$ for all $r\in \reg$. 
For every data word $w=(a_1,d_1)\dots(a_n,d_n)$, define $\pi(w)=(a_1,\pi(d_1))\dots(a_n,\pi(d_n))$. 
Note that the application of $\pi$ on $q$ and $w$ preserves the reachability property, \ie, 
$\post(\pi(q),\pi(w))=\{\pi(q')\mid q'\in \post(q,w)\}$.

\bigskip

Assuming that
$\R$ has some synchronizing data word, we first prove the following claim by an induction.

\medskip
\noindent {\bf Claim.} For  all  pairs of configurations~$q_1,q_2$,
if there exists~$w$  such that $\abs{\post(\{q_1,q_2\},w)}=1$,
then 
\begin{itemize}
	\item 	
		for all sets~$X=\{x_1,x_2,\cdots,x_{2k+1}\}\subset \domain$ with  $\data(q_1),\data(q_2)\subseteq X$, 
	\item there exists some data word~$w_{q_1,q_2}\in (\Sigma\times X)^*$  
	such that $\abs{\post(\{q_1,q_2\},w_{q_1,q_2})}=1$.
\end{itemize}

\noindent Note that by $\abs{X}=2k+1$, the data efficiency of $w_{q_1,q_2}$ is  at most~$2k+1$. 

\medskip
\noindent {\bf Proof of Claim.}  
Let $q_1$ and $q_2$ be two configurations of~$\R$ and define $\data(q_1,q_2)=\data(q_1)\cup \data(q_2)$. Since 
$\R$ has some synchronizing data words,  there exists~$w$ such that $\abs{\post(\{q_1,q_2\},w)}=1$.
The proof is by an induction on the length of~$w$.

\noindent {\bf Base of induction.}  
Assume~$w=(a,d)$ have length~$\abs{w}=1$.
Let~$X$ be any arbitrary set of data such that $\abs{X}=2k+1$ and 
$\data(q_1,q_2)\subseteq X$. 
There are two cases:
\begin{itemize}
	\item  $d\in X$: This entails that $\data(w) \subseteq X$.
	Observe that $w_{q_1,q_2}=w$ satisfies the induction statement. 
	\item $d \not \in X$:
	Since $\abs{\data(q_1,q_2)}\leq 2k$, there exists data~$x\neq d$ such that
  $x=X \setminus\data(q_1,q_2)$.
	Since $x\neq d$, we can define the bijection $\pi:\{d\}\cup \data(q_1,q_2)\to \{x\}\cup \data(q_1,q_2)$
	such that $\pi(d)=x$ and $\pi(d')=d'$ for all $d'\in \data(q_1,q_2)$.
	Observe that $\pi(q_i)=q_i$ for all~$i\in \{1,2\}$. 
	Then  
	$$\abs{\post(\{q_1,q_2\},(a,d))}=\abs{\post(\{\pi(q_1),\pi(q_2)\},(a,\pi(d))} = \abs{\post(\{q_1,q_2\},(a,x))}.$$
	This and the assumption $\abs{\post(\{q_1,q_2\},(a,d))}=1$ yield 
	$\abs{\post(\{q_1,q_2\},(a,x))}=1$. 
	The word~$w_{q_1,q_2}=(a,x)$ satisfies the induction statement.

\end{itemize}
The base of induction hence holds.

\noindent {\bf Step of induction.} Assume that  the induction hypothesis holds for $i-1$.
Consider some word~$(a,d)\cdot w$
such that $\abs{w}=i-1$ and $\abs{\post(\{q_1,q_2\},(a,d)\cdot w)}=1$.

Consider some set~$X$ which has cardinality~$2k+1$ and $\data(q_1,q_2)\subseteq X$,
we construct the data word~$w_{q_1,q_2}$ as follows.
Let $p_1=\post(q_1,(a,d))$ and $p_2=\post(q_2,(a,d))$, and
let $\data(p_1,p_2)=\data(p_1) \cup \data(p_2)$.
Due  to the fact that $p_1,p_2$ are successors of $q_1,q_2$ after inputting~$(a,d)$,
we know that if $d\in \data(q_1,q_2)$ then  $d\in \data(p_1,p_2)$.
There are two cases:
\begin{itemize}

	\item $d\in \data(q_1,q_2)$ or $d\not \in \data(p_1,p_2)$. 
	 These guarantee that $\data(p_1,p_2) \subseteq \data(q_1,q_2)$ if
	$d\in \data(q_1,q_2)$, and that $\data(p_1,p_2)= \data(q_1,q_2)$ if 
	$d\not \in \data(p_1,p_2)$.
  As a result, $\data(p_1,p_2) \subseteq X$.
	By induction hypothesis, there exists some data word~$w_{p_1,p_2}$ over data domain~$X$ 
	such that $\abs{\post(\{p_1,p_2\},w_{p_1,p_2})}=1$.
	For $w_{q_1,q_2}=(a,d) \cdot w_{p_1,p_2}$ the   statement of induction holds, as  
	$\abs{\post(\{q_1,q_2\},w_{q_1,q_2})}=1$.

	\item  $d\not \in \data(q_1,q_2)$ and $d\in \data(p_1,p_2)$. 
	Without loss of generality, we assume that~$d\not \in X$.
	Otherwise $d\in X$ would imply $\data(p_1,p_2)\subseteq X$, and we simply let $w_{q_1,q_2}=w_{p_1,p_2}$. 
	Since $\abs{\data(q_1,q_2)}\leq 2k$, there exists some datum~$x\neq d$ such that
  $x\in X \setminus\data(q_1,q_2)$.
	Since $x\neq d$, we can define the bijection $\pi:\{d\}\cup \data(q_1,q_2)\to \{x\}\cup \data(q_1,q_2)$
	such that $\pi(d)=x$ and $\pi(d')=d'$ for all $d'\in \data(q_1,q_2)$.
	Since $\data(p_1,p_2)\setminus \{d\} \subseteq \data(q_1,q_2)$, having~$d$ in the domain of~$\pi$,
	the bijection~$\pi$ ranges over~$\data(p_1,p_2)$.
	By induction hypothesis, there exists some data word~$w_{p_1,p_2}$ over data domain~$(X\setminus\{x\})\cup(\{d\})$ 
	such that $\abs{\post(\{p_1,p_2\},w_{p_1,p_2})}=1$.
  Then, $\abs{\post(\{\pi(p_1),\pi(p_2)\},\pi(w_{p_1,p_2}))}=1$.
	For all $1\leq i\leq 2$,  we have  $\pi(p_i)\in \post(q_i,(a,x))$ 
	since $p_i\in \post(q_i,(a,d))$ and $x=\pi(d)$.
	By above arguments, we conclude that  $\abs{\post((\{q_1,q_2\},(a,x)\pi(w_{p_1,p_2})}=1$.
	As  $\{x\}\cup \data(\{q_1,q_2\})\subseteq X$, thus the data word~$w_{q_1,q_2}=(a,x)\pi(w_{p_1,p_2})$ satisfies the 
	statement of induction.
\end{itemize}
 
The above arguments prove that in all cases, there exists~$w_{q_1,q_2}\in (\alphabet \times X)^*$ that merges two configurations
$q_1$  and $q_2$ into a singleton, which completes the proof of {\bf Claim}.

\bigskip
\bigskip
Since~$\R$ has some synchronizing data word, using Lemma~\ref{lemmafiniteD},
we know that there exists some word~$w$ with data efficiency~$k$ 
such that $\post(\locs\times \domain^k, w)\subseteq \locs \times \data(w)^k$.
Consider some set~$X=\{x_1,x_2,\cdots,x_{2k+1}\}\subset \domain$ such that $\data(w)\subseteq X$.
We use the pairwise synchronization technique as follows. 
Define $S_{n}=\locs \times X^k$ and $n=\abs{\locs}(2k+1)^k$, \ie, $\abs{S_n}=n$. 
For all $i=n-1,\cdots,1$ repeat the following:
\begin{enumerate}
\item Take a pair of configurations~$q_1,q_2\in S_{i+1}$. By the~{\bf Claim} above, one can find 
	some word~${w_{q_1,q_2}\in (\alphabet\times X)^{*}}$ such that $\abs{\post(\{q_1,q_2\},w_{q_1,q_2})}=1$,
\item Define $v_i= w_{q_1,q_2}$ and $S_{i}=\post(S_{i+1},v_i)$.
\end{enumerate}
Note that by determinism of $\R$, for every $i\in\{1,\cdots, n-1,\}$, we have $\abs{S_{i}}\le \abs{S_{i+1}}-1$. 
Thus the word~$w_{\synch}=w\cdot v_{n-1} \cdots v_2 \cdot v_1$ is a synchronizing data word for~$\R$.
Since $\data(w)\subseteq X$ and $\data(v_i)\subseteq X$ for all $i\in \{1,\cdots,n-1\}$, 
the data efficiency of $w_{\synch}$ is at most $2k+1$. The proof is complete.
\end{proof}

\begin{qlemma}{\ref{theokDRA}}
\ \theokDRA
\end{qlemma}
\begin{proof}
	We prove $\pspace$-hardness by a reduction from the non-emptiness problem for $k$-DRA. 
	Let $\R=(\locs,\reg,\Sigma,T)$ be a $k$-DRA equipped with an initial location $\loc_i$ and an accepting location $\loc_f$, where, without loss of generality, we assume that all outgoing transitions from $\loc_i$ update all registers, and that $\loc_f$ has no outgoing edges. We also assume that $\R$ is complete, otherwise, we add some non-accepting location and direct  all undefined transitions to it.

	The reduction is such that from $\R$ we construct another $k$-DRA $\R_{\mathsf{syn}}$ such that the language of $\R$ is not empty if, and only if, $\R_{\mathsf{syn}}$ has some synchronizing data word. 
	We define $\R_{\mathsf{syn}}=(\locs_{\mathsf{syn}},\reg,\Sigma_{\mathsf{syn}},T_{\mathsf{syn}})$ as follows. 
	The set of locations is $\locs_{\mathsf{syn}}=\locs\cup\{\lreset\}$, where $\lreset\not\in \locs$ is a new location;
	the alphabet is $\Sigma_{\mathsf{syn}}=\Sigma\cup\{\star\}$, where $\star\not\in\Sigma$. 
	To define $T_{\mathsf{syn}}$, we add the following transitions to $T$. 
	\begin{itemize}
	\item $\loc_f \ttto{}{a}{\reg} \loc_f$ for all letters $a\in\alphabet_{\mathsf{syn}}$,
	\item $\loc_i \ttto{}{\star}{\reg} \loc_i$ 
	\item $\lreset \ttto{}{a}{\reg} \loc_i$ for all letters $a\in\alphabet_{\mathsf{syn}}$,
	\item $\loc \ttto{}{\star}{\reg} \lreset$ for all  $\loc\in\locs_{\mathsf{syn}}$ except for $\lreset,\loc_i,\loc_f$. 
	\end{itemize}
	Note that $\R_{\mathsf{synch}}$ is indeed deterministic and complete. 
	To establish the correctness of the reduction, we prove that the language of $\R$ is not empty if, and only if, $\R_{\mathsf{syn}}$ has a synchronizing data word. 
	
	First, assume that the language of $\R$ is not empty. 
	Then there exists a data word $w=(a_1,d_1)\dots(a_n,d_n)$ such that $w\in L(\R)$. 
	Hence there exists a run starting from $(\loc_i,\val_i)$ and ending in $(\loc_f,\val_f)$ for some $\nu_i,\nu_f\in \domain^{\abs{\reg}}$. 
	The data word $(\star,d)(\star,d) w (\star,d)$ for some $d\in\domain$ synchronizes $\R_{\mathsf{syn}}$ in location $\loc_f$.

	Second, assume that $\R_{\mathsf{syn}}$ has some synchronizing data word. Let $w\in(\Sigma_{\mathsf{syn}}\times\domain)^*$ be one of the shortest data synchronizing data words. 
	All transitions in $\loc_f$ are self-loops with update on all registers; Hence, $\R_{\mathsf{syn}}$ can only be synchronized in $\loc_f$. 
	Hence, we also have $\post((\loc_i,\nu_i),w)=\{(\loc_f,\nu_f)\}$ (for some $\nu_i,\nu_f\in\domain^{\abs{\reg}}$). 
	By the fact that $w$ is a shortest synchronizing data word, we can infer that the corresponding run does not contain any $\star$-transitions except for two self-loops in $\loc_i$ in the very beginning.	
	Hence there exists a run from $(\loc_i,\nu_i)$ to $\loc_f$ and thus $L(\R)\neq\emptyset$.
\end{proof}
\section{Proofs for Non-deterministic Register Automata}
\label{append_NRAs}
\begin{qlemma}{\ref{lemmacountingNRA}}
\ \lemmacountingNRA
\end{qlemma}
\begin{proof}
	The family of $1$-NRAs~$(\R_{\counter(n)})_{n\in \nat}$ is defined as follows. 
We define the alphabet of RA~$\R_{\counter(n)}$ by~$\alphabet=\{\#,\star, \Bit{0},\Bit{1},\cdots,\Bit{n}\}$.
The structure of~$\R_{\counter(n)}$ is composed of three distinguished locations~$\synch,\lreset,\lzero$ and  
locations $\loctwo{n},\loctwo{n-1},\cdots,\loctwo{1},\loctwo{0}$ and $\locctwo{n},\locctwo{n-1},\cdots,\locctwo{1},\locctwo{0}$.
The general structure of $\R_{\counter(n)}$ is partially depicted in \figurename~\ref{fig:countingNRA}. 
The RA~$\R_{\counter(n)}$ is constructed such that for all synchronizing data words~$w$,  
some datum~$x\in \data(w)$  appears in~$w$ at least $2^n$~times.
A counting feature is thus embedded in~$\R_{\counter(n)}$: intuitively,
the set of all reached configurations represents the counter value.
Starting from $\{(\lzero,x)\}$, 
the first increment results in $\{\locctwo{n},\cdots,\locctwo{2},\locctwo{1},\loctwo{0}\}\times\{x\}$,
where location~$\loctwo{i}$ means that the $i$-th least significant bit in the  binary representation of the counter value is set to~$\inBinary{1}$, and  
location~$\locctwo{i}$ means that the $i$-th bit is set to~$\inBinary{0}$.
Informally, we say that there is an $x$-token in every reached location. Here,  $\locctwo{n},\cdots,\locctwo{2},\locctwo{1},\loctwo{0}$ have
$x$-tokens.
A sequence of counter increments is encoded by re-placing the $x$-tokens, as shown  in the following sequence of sets of locations: 
$\{\locctwo{n}, \cdots,\locctwo{2},\loctwo{1},\locctwo{0}\}$,  
$\{\locctwo{n},\cdots,\locctwo{2},\loctwo{1},\loctwo{0}\}$,  
$\{\locctwo{n},\cdots,\locctwo{3},\loctwo{2},\locctwo{1},\locctwo{0}\}$, etc. 
The transitions of $\R_{\counter(n)}$ are defined in such a way that, starting from $\{(\lzero,x)\}$,   \emph{either} $\loctwo{i}$ \emph{or} 
$\locctwo{i}$ have tokens, but never both of them at the same time. We now present a detailed explanation of the structure of~$\R_{\counter(n)}$.

All transitions in~$\synch$ are self-loops with an update on the register $\synch \ttto{}{\alphabet}{r} \synch$. 
Thus,  $\R_{\counter(n)}$ can only be synchronized in~$\synch$. 
Moreover, $\synch$ is only accessible by $\#$-transitions. 
Similarly, all transitions except for those with label $\star$, are self-loops in location $\lreset$; thus, $\R_{\counter(n)}$ can only be synchronized by leaving $\lreset$ by reading $\star$. 
We use this also to avoid transitions which are \emph{incorrect} with respect to the binary incrementing process: all incorrect actions are guided to $\lreset$ to enforce another~$\star$. 
Assuming~$w$ to be one of the shortest synchronizing words, we see that
$\post(\locs\times \domain,w)=\{(\synch,x)\}$, where $w$ starts with $(\star,x)$ and ends with $(\#,x)$.

The counting  involves an \emph{initializing process} and several \emph{incrementing} processes.
\begin{itemize}
\item \emph{initializing the counter to $\lzero$}:
	the $\star$-transitions are devised to place a token in~$\lzero$: 
	from all locations $\ell\in \locs \setminus\{\synch\}$ we have
	$\ell  \ttto{}{\star}{r} \lzero$. This sets the counter to~$\inBinary{0}$.
	
\item \emph{incrementing the counter}: we use $\Bit{0},\dots,\Bit{n}$-transitions with equality guards to control the
	increment. 
	Intuitively, an equality-guarded $\Bit{i}$-transition is taken to set the $i$-th bit in the binary representation of the counter value according
  to the standard rules of binary incrementation.
	\\
	Initially, the token in $\lzero$ splits in $\loctwo{0}$ and $\locctwo{n},\cdots \locctwo{1}$ to represent $\inBinary{0}\cdots\inBinary{01}$, 
	by taking the transitions
	$\lzero \tto{=r}{\Bit{0}} \loctwo{0}$ and $\lzero \tto{=r}{\Bit{0}} \locctwo{j}$ for all $1\leq j\leq n$.  
	Equality-guarded $\Bit{i}$-transitions for $i\in\{1,\dots,n\}$ are incorrect for $\lzero$ and thus guided to $\lreset$.
	Whenever data different from $x$ is processed, $\R_{\counter(n)}$ takes self-loops (omitted in \figurename~\ref{fig:countingNRA}) 
	and keeps the $x$-tokens unmoved. 
	\\
	The equality-guarded $\Bit{i}$-transitions should only be taken if the $i$-th bit is not set, or, equivalently, 
	if  the location $\loctwo{i}$ contains no token.
	This is guaranteed by a $\Bit{i}$-transition  $\loctwo{i} \tto{=r}{\Bit{i}} \lreset$, for every $0\leq i\le n$, which results in an incorrect transition and should be avoided. (Otherwise the counting process has to restart from~$\inBinary{0}$.) In \figurename~\ref{fig:countingNRA}, we depict the corresponding transitions for $i=2$ and $i=n$. 
	\\
	Further,  we need to guarantee that for all $i\geq 1$ a $\Bit{i}$-transition is taken only if all less significant bits are set, 
	or, equivalently, if all locations~$\loctwo{i-1},\cdots \loctwo{0}$ contain a token. 
	This is ensured by a $\Bit{i}$-transition $\locctwo{j} \tto{=r}{\Bit{i}} \lreset$, for every $0\leq j<i$, which again results in an incorrect transition. See, \eg, the transition $\locctwo{2}\tto{=r}{\Bit{i}}\lreset$ in \figurename~\ref{fig:countingNRA} for every $3\le i\le n$.

	Finally, $\Bit{i}$-transitions  must produce tokens in $\loctwo{i}$ and $\locctwo{0},\cdots \locctwo{i-1}$,
	thus $\locctwo{i} \tto{=r}{\Bit{i}} \loctwo{i}$ and $\loctwo{j} \tto{=r}{\Bit{i}} \locctwo{j}$  for all $0\leq j< i$. All tokens in locations $\loctwo{j}$ and $\locctwo{j}$, respectively, for $j>i$ remain where they are, which is implemented by equality-guarded $\Bit{i}$-self-loops in  $\loctwo{j}$ and $\locctwo{j}$, respectively.
\end{itemize}  
By construction, it is easy to see that $\Bit{i}$-transitions are the only way to produce a token in~$\loctwo{i}$,
which can be fired if $\locctwo{i}$ has a token. The  $\Bit{i}$-transitions then consume the token in 
$\locctwo{i}$. 
This guarantees that after the first $\star$-transition, which puts a token into  $\lzero$, 
the two locations $\loctwo{i}$ and $\locctwo{i}$ will never have a token at the same time.

Finally, all equality-guarded $\#$-transitions in $\locctwo{n}$ and $\loctwo{i}$ for all $0\le i<n$ are sent to $\lreset$. In contrast, all $\#$-transitions in $\loctwo{n}$ and $\locctwo{i}$ for all $0\le i<n$ are sent to $\synch$, with an update on the register.  
This guarantees that the counter must correctly count from~$\inBinary{0}$ to $\inBinary{1}\inBinary{0}\cdots \inBinary{0}$, meaning that at least one
datum~$x$ appears at least $2^n$ times while synchronizing~$\R_{\counter(n)}$. 
\end{proof}

\begin{qlemma}{\ref{lemma_nonuniv_to_synch}}
\ \lemmaNonunivToSynch
\end{qlemma}
\begin{proof}
The reduction is based on the    construction presented in Theorem 17 in~\cite{DBLP:conf/fsttcs/0001JLMS14}.

Let $\R =\tuple{\locs,\reg,\alphabet,T}$ be an NRA equipped  with 
an initial location~$\loc_\is$ and a set $\locs_\acc$ of accepting locations,
where, without loss of generality, we assume that all outgoing transitions from~$\loc_\is$ update all registers.	
We also assume that $\R$ is complete, otherwise, we add some non-accepting location and direct all undefined transitions
 to it.

 We construct an NRA~$\R_{\mathsf{syn}}$ such that there exists some data word that is not in $L(\R)$  if, and only if,~$\R_{\mathsf{syn}}$ has some synchronizing data word. 
We define $\R_{\mathsf{syn}} =\tuple{\locs_{\mathsf{syn}},\reg,\alphabet_{\mathsf{syn}},T_{\mathsf{syn}}}$ as follows.
The set of locations is $\locs_{\mathsf{syn}}=\locs\cup\{\lreset,\synch\}$ where $\synch,\lreset \not\in\locs$ are two new locations. 
The alphabet is $\alphabet_{\synch}=\alphabet\cup\{\#,\star\}$ where $\#,\star\not\in\alphabet$. 
The transition relation~$T_{\mathsf{syn}}$ is the union of $T$ and set containing the following transitions: 
\begin{itemize}
			\item $\synch \ttto{}{a}{\reg} \synch$ for all letters $a\in\alphabet_{\mathsf{syn}}$,
			\item $\lreset \ttto{}{\star}{\reg} \loc_\is$ and $\lreset \ttto{}{a}{\reg} \lreset$ for all letters $a\in\alphabet_{\mathsf{syn}}\backslash\{\star\}$,
			\item $\loc \ttto{}{\star}{\reg} \loc_\is$  for all locations~$\loc \in \locs$, 
			\item $\loc\ttto{}{\#}{\reg}\synch$ for all non-accepting locations~$\loc\in\locs\backslash\locs_\acc$, 
			\item $\loc\ttto{}{\#}{\reg} \lreset$ for all accepting  		locations~$\loc\in\locs_\acc$. 
\end{itemize}

Next, we prove the correctness of the reduction. 
	
First, assume there exists a data word $w=(a_1,d_1)\dots(a_n,d_n)$ such that $w\not\in L(\R)$. 
Hence, all runs starting in~$(\loc_\is,\valuation_i)$  with  $\valuation_i\in \domain^{\abs{\reg}}$ 
end in some configuration $(\loc,\val)$ with $\loc\not\in\locs_\acc$. 
The data word $(\star,d) \cdot w\cdot (\#,d)$ with~$d\in \domain$ synchronizes~$\R_{\mathsf{syn}}$ in location~$\synch$, proving that
$\R_{\mathsf{syn}}$ has some synchronizing data word.

Second, assume that $\R_{\mathsf{syn}}$ has some synchronizing data word. 
All transitions in~$\synch$ are self-loops with update on all registers;
thus,  $\R_{\mathsf{syn}}$ can only synchronize in~$\synch$. 
Moreover, $\synch$ is only accessible with $\#$-transitions;
assuming~$w$ is one of the shortest synchronizing data words, we see that
$\post(\locs\times \domain,w)=\{(\synch,\valuation))\}$ for some $\valuation\in \domain^{\abs{\reg}}$.
From all locations~$\ell\in \locs$ we have
$\ell  \ttto{}{\star}{\reg} \loc_\is$;
we say that $\star$-transitions \emph{reset}~$\R_{\mathsf{syn}}$.
Moreover, the only outgoing transition in location~$\lreset$ is the $\star$-transition.
Thus, a \emph{reset} followed by some $\#$  must occur while synchronizing. 
Let $w=w_0(\star,d_\star) w_1 (\#,d_{\#}) w_2$, where $w_1\in(\Sigma\times \domain)^+$ 
is the data word between the last occurrence of $\star$ and the first following occurrence of $\#$, 
and $w_2\in(\Sigma'\backslash\{\star\})^*$.
We prove that $w_1\not\in L(\R)$.  
By contradiction, assume that~$w_1$ is in the language; thus, 
there exist valuations~$\valuation_i,\valuation_f\in \domain^{{\abs{\reg}}}$
such that $\R_{\mathsf{syn}}$ has a run over~$w_1$, \ie, starting in~$(\loc_\is,\valuation_i)$
 and ending in~$(\loc_f,\valuation_f)$ where~$\loc_f\in \locs_\acc$.
In fact, since all outgoing transitions in~$\loc_\is$ update all registers, then
for all valuations~$\valuation_i$,  $\R_{\mathsf{syn}}$ has an accepting run over~$w_1$. 

Note that $w_0$ cannot be a synchronizing word for $\R_{\mathsf{syn}}$, because this would contradict 
the assumption that $w$ is one of the shortest synchronizing data word. 
It implies that there must be some configuration~$q$ such that $\post_{\R_{\mathsf{syn}}}(q,w_0)$ contains 
some configuration~$(\loc,\valuation)$ 
with $\loc\neq \synch$.
From~$(\loc,\valuation)$, inputting  the next~$(\star,d_\star)$ (that is after~$w_0$ in synchronizing word~$w$), we reach 
$(\loc_\is,\{d_\star\}^{\abs{\reg}})$.
Since for all valuations~$\valuation_i$,  starting in~$(\loc_\is,\valuation_i)$, $\R_{\synch}$ has an accepting run over~$w_1$, 
it must have an accepting run from $(\loc_\is,\{d_\star\}^{\abs{\reg}})$ to some accepting configuration~$(\loc_f,\valuation_f)$ too.
Reading  the last~$\#$ (that is after~$w_1$ in synchronizing word~$w$),  $\lreset$ is reached. 
Since $w_2$ does not contain any $\star$, $\lreset$ is never left, meaning that  $\R_{\mathsf{syn}}$ cannot synchronize in $\synch$, a contradiction. 	
The proof is complete.	

Note that the reduction preserves the number of registers in the NRAs. 
\end{proof}
\end{document}